\documentclass[sigconf,nonacm]{acmart}

\usepackage{amsmath}
\usepackage{microtype}
\usepackage{booktabs}
\usepackage{algorithm}
\usepackage{algpseudocode}
\usepackage{enumitem}
\usepackage{tikz}
\usepackage{fontawesome5}
\usetikzlibrary{arrows.meta,positioning,fit,calc,backgrounds,decorations.pathreplacing,shapes.geometric}
\usepackage{pgfplots}
\pgfplotsset{compat=1.18}
\usepackage{xspace}
\usepackage{framed}
\usepackage[most]{tcolorbox}
\usepackage{subcaption}
\usepackage{needspace}
\usepackage{dblfloatfix}
\usepackage{placeins}
\usepackage{cuted}

\definecolor{cBlue}{RGB}{65,105,225}
\definecolor{cOrange}{RGB}{230,145,56}
\definecolor{cTeal}{RGB}{0,128,128}
\definecolor{cRed}{RGB}{190,50,50}
\definecolor{cGreen}{RGB}{3,141,89}
\definecolor{cPurple}{RGB}{130,80,180}
\definecolor{cGray}{RGB}{120,120,120}

\newtcolorbox{propbox}{
  enhanced,
  colback=cBlue!4,
  colframe=cBlue!72!black,
  arc=2.5pt,
  boxrule=0.5pt,
  left=10pt, right=10pt, top=6pt, bottom=6pt,
  before skip=6pt, after skip=6pt,
  borderline west={1.8pt}{0pt}{cBlue!72!black},
}

\newtcolorbox{rqbox}{
  enhanced jigsaw, breakable=false,
  colback=black!7, colframe=black!7,
  boxrule=0pt, arc=1pt,
  left=6pt, right=6pt, top=6pt, bottom=6pt,
  before skip=\topsep, after skip=\topsep,
}

\pgfplotsset{
  every axis/.append style={
    grid style={dashed, gray!20},
    tick label style={font=\scriptsize},
    label style={font=\small\sffamily},
    legend style={
      font=\scriptsize,
      fill=white!90!gray,
      draw=none,
      rounded corners=2pt,
      legend cell align=left,
    },
    axis line style={black, line width=0.3pt},
  },
}

\setlist[itemize]{leftmargin=*,topsep=2pt,itemsep=1.5pt,parsep=0pt}
\setlist[enumerate]{leftmargin=*,topsep=2pt,itemsep=1.5pt,parsep=0pt}

\title{Synthesizing Multi-Agent Harnesses for Vulnerability Discovery}

\author{Hanzhi Liu}
\affiliation{%
  \institution{University of California, Santa Barbara}
  \country{}
}
\email{hanzhi@ucsb.edu}

\author{Chaofan Shou}
\affiliation{%
  \institution{Fuzzland}
  \country{}
}
\email{shou@fuzz.land}

\author{Xiaonan Liu}
\affiliation{%
  \institution{Fuzzland}
  \country{}
}
\email{xl@fuzz.land}

\author{Hongbo Wen}
\affiliation{%
  \institution{University of California, Santa Barbara}
  \country{}
}
\email{hongbowen@ucsb.edu}

\author{Yanju Chen}
\affiliation{%
  \institution{University of California, San Diego}
  \country{}
}
\email{yanju@ucsd.edu}

\author{Ryan Jingyang Fang}
\affiliation{%
  \institution{World Liberty Financial}
  \country{}
}
\email{ryan@worldlibertyfinancial.com}

\author{Yu Feng}
\affiliation{%
  \institution{University of California, Santa Barbara}
  \country{}
}
\email{yufeng@cs.ucsb.edu}

\begin{document}

\begin{abstract}
LLM agents have begun to find real security vulnerabilities
that human auditors and automated fuzzers missed for decades,
in source-available targets where the analyst can build and
instrument the code.
In practice the work is split among several agents, wired
together by a \emph{harness}: the program that fixes which
roles exist, how they pass information, which tools each may
call, and how retries are coordinated.
When the language model is held fixed, changing only the
harness can still change success rates by several-fold on
public agent benchmarks, yet most harnesses are written by
hand; recent harness optimizers each search only a narrow
slice of the design space and rely on coarse pass/fail
feedback that gives no diagnostic signal about \emph{why} a
trial failed.
\textsc{AgentFlow} addresses both limitations with a typed
graph DSL whose search space jointly covers agent roles,
prompts, tools, communication topology, and coordination
protocol, paired with a feedback-driven outer loop that reads
runtime signals from the target program itself to diagnose
which part of the harness caused the failure and rewrite it
accordingly.
We evaluate \textsc{AgentFlow} on TerminalBench-2 with
Claude~Opus~4.6 and on Google~Chrome with Kimi~K2.5.
\textsc{AgentFlow} reaches \textbf{84.3\%} on TerminalBench-2,
the highest score in the public leaderboard snapshot we
evaluate against, and discovers \textbf{ten previously unknown zero-day
vulnerabilities} in Google~Chrome, including \textbf{two Critical
sandbox-escape vulnerabilities} (CVE-2026-5280 and
CVE-2026-6297).
\end{abstract}

\maketitle

\section{Introduction}
\label{sec:intro}

Language-model agents have begun to find real security
vulnerabilities that human auditors and automated fuzzers missed for
decades.
Google's Big Sleep agent discovered a zero-day memory-corruption
flaw in SQLite (CVE-2025-6965) before attackers could exploit
it~\cite{google2024bigsleep}.
Anthropic's Claude Mythos Preview found a 27-year-old denial-of-service
vulnerability in OpenBSD's TCP SACK implementation, crashing any
OpenBSD host with two
packets~\cite{anthropic2026mythos}.
LLM agents have autonomously hacked real-world
websites~\cite{fang2024llmagent}; teams of LLM agents have further
improved performance on real-world zero-day web-vulnerability
exploitation~\cite{fang2024teams}; and agentic systems have
reproduced CTF-level exploits from natural-language
descriptions~\cite{shao2024nyu}.
The capability is real and accelerating.

To see where these capabilities come from, and why they are still
limited, it helps to look at how an LLM-based vulnerability
finder is actually put together.
In the simplest deployment, a single language model operates as a
\emph{security agent}: the operator hands it a \emph{target
program} (the software under analysis, e.g.\ \texttt{libtiff},
OpenSSL, \texttt{curl}), gives it a \emph{system prompt} that
states its goal in natural language (``find a memory-safety
vulnerability in the TIFF parsing module''), and exposes a set
of \emph{tools}: concrete actions the model may invoke at each
step (read source files, compile the target with
instrumentation, execute the binary, inspect output).
The agent reasons about the target, generates candidate inputs,
runs them, and observes the results.
This loop of reasoning, tool invocation, and feedback is the
single-agent system.

A single agent has to do everything inside one reasoning trace:
read the target's validation logic, craft structured inputs that
satisfy format checks, run the target, and decide whether what
came back is a real crash.
This breaks down for three reasons.
First, real security targets produce voluminous output: a single
instrumented Chrome build emits megabytes of sanitizer and
coverage data per run, quickly exhausting even frontier-scale
context windows (on the order of $200$K tokens).
Second, when heterogeneous sub-tasks (source analysis,
input crafting, crash triage) compete for the same context,
the model exhibits \emph{lost-in-the-middle}
effects~\cite{liu2024lost}: it drops earlier analysis, repeats
work it already did, and abandons long-horizon strategies.
Third, a single trace cannot explore multiple hypotheses in
parallel; it must commit to one strategy at a time, and a
dead-end wastes the entire budget.
These limitations are well documented in the agent-systems
literature~\cite{hong2024metagpt,wu2023autogen,hu2024adas}
and are the reason that current state-of-the-art systems do not
rely on a single agent.

Production systems instead split the work across specialized
agents, each with its own prompt, LLM model (the
underlying language model assigned to that role), and tools,
much like a small security team: an \emph{analyst} extracts the
preconditions a valid input must satisfy from the source, an
\emph{explorer} crafts inputs guided by the analyst's summary,
and a \emph{verifier} runs the candidate and decides whether
the resulting behavior counts as a real crash, with the
verifier's verdict feeding back to the analyst on failure.
A team of three is strictly more powerful than a soloist, but
it adds a new responsibility: someone has to decide
\emph{which} roles exist, \emph{who} talks to \emph{whom},
\emph{what} information passes along each edge, and
\emph{when} the team retries.

The agent-systems community refers to the orchestration code that
makes these decisions as the
\emph{harness}~\cite{google2024bigsleep,anthropic2026mythos,lee2026metaharness}:
the program that specifies which agents exist, what prompt each
one receives, what tools each may invoke, how outputs are routed
through the communication graph, and what coordination protocol
governs execution (sequential, parallel, fan-out, retry on
failure).
Designing a good harness is the central concern of this paper.
On the public TerminalBench-2 leaderboard~\cite{terminalbench2},
three systems run the same Claude~Opus~4.6 model on the
same $89$ long-horizon tasks but their pass rates span a
$4\times$ range, from $20\%$ to $80\%$~\cite{lee2026metaharness};
with the model held constant, the design choices that remain
(orchestration, prompts, tools, coordination, and feedback
channels) account for the spread.
A well-designed harness solves more tasks with fewer
inference calls; a poorly-designed one burns compute on
dead-end strategies and runs into the tens of thousands of
dollars per
campaign~\cite{google2024bigsleep,anthropic2026mythos}.

Because harness engineering matters this much, recent work has
begun to automate it: instead of hand-tuning the harness, an
\emph{outer-loop optimizer} runs the current harness, observes
how it did, and proposes a new harness for the next trial.
Meta-Harness~\cite{lee2026metaharness} rewrites that agent's
prompts, tool bindings, and context-building code but keeps the
team to a single agent.
ADAS~\cite{hu2024adas} introduces new agent code (new roles with
new prompts and tool calls) but keeps the communication graph
between agents fixed by a hand-written controller.
AFlow~\cite{zhang2024aflow} uses a tree search over a fixed library
of predefined workflow operators, restructuring how agents are wired
together but holding the agent pool, the prompts, and the tools
constant.
MaAS~\cite{yuan2025maas} samples agent teams from a pre-defined pool
but fixes how those agents communicate and hand off control.
Each system improves over hand-designed harnesses on
its own benchmark, yet they share two fundamental limitations.

\emph{Narrow scope.}
Each optimizer searches only a small slice of the harness design
space: Meta-Harness edits prompts and tool bindings within a
single agent; ADAS generates new Python controller code but fixes
the communication graph; AFlow picks from a library of predefined
operators; MaAS samples agents but fixes their communication
protocol.
The reason is tractability: opening the search space to all
harness components simultaneously makes na\"ive search
intractable.
In practice this means they cannot discover solutions that require
cross-component changes (e.g., adding a new agent \emph{and}
rewiring the communication graph \emph{and} adjusting its prompt
in a single step), which our evaluation shows are precisely the
edits that separate competitive harnesses from mediocre ones.

\emph{Coarse feedback.}
Existing optimizers typically rely on raw agent traces or binary
pass/fail outcomes as the signal for the next proposal.
This ``zero or one'' feedback provides no fine-grained diagnostic
information about \emph{why} a trial failed: whether the
agent's input never reached the vulnerable function, whether
the crash was masked by an unrelated assertion, or whether the
target's defense was never exercised.
In a large search space, coarse feedback turns the optimization
into a random walk.

These two limitations motivate \textsc{AgentFlow}.

\emph{Addressing narrow scope with a typed graph DSL.}
To search over all harness components at once,
\textsc{AgentFlow} represents every harness as a program in a
typed graph DSL: nodes are agents, edges are dataflow or retry
links, and all harness dimensions (agent roles
$\mathcal{A}$, communication topology $\mathcal{G}$,
message schemas $\Sigma$, tool bindings $\Phi$, and
coordination protocol $\Psi$) each map to a first-class
editable field in the program.
A single optimization step can therefore add a new agent,
rewire the communication graph, adjust the new agent's prompt,
and restrict its tool set, all as one local rewrite.
The obvious risk is that opening the full
space makes search intractable.
\textsc{AgentFlow} controls this explosion through a type
system over the graph: before any candidate harness is
dispatched for the expensive LLM evaluation, a
well-formedness check verifies that every template variable
resolves to a declared agent output or feedback channel, that
every edge connects declared agents, and that the graph is
connected.
The type system is what makes searching the full search space
tractable: it eliminates structurally broken candidates
cheaply, so the \textsc{AgentFlow} spends its budget only on
well-formed harnesses.

\emph{Addressing coarse feedback with runtime diagnostics.}
Instead of a binary pass/fail signal,
\textsc{AgentFlow} reads structured runtime feedback from the
target program itself (test verdicts, line-level coverage
maps, sanitizer reports), together with each agent's action
traces and the archive of all prior trials.
This information lets the system localize \emph{why} a
harness failed rather than only \emph{that} it failed.
For example, a coverage map indicates whether an agent's
input reached the vulnerable function; a sanitizer trace
distinguishes a benign crash from the memory-safety
vulnerability the campaign is targeting; and the archive of
prior trials identifies when a proposed harness repeats one
that was previously evaluated.
Each of these signals provides a concrete lever that a purely
binary verdict does not.

We evaluate \textsc{AgentFlow} on two workloads that differ
in domain, scale, and LLM model.
On \textbf{TerminalBench-2}~\cite{terminalbench2}, a public
leaderboard of $89$ long-horizon terminal tasks, the
synthesized harness reaches \textbf{84.3\%}
with Claude~Opus~4.6~\cite{anthropic2025opus}, the highest
score among all publicly-ranked harnesses.
On the \textbf{Google~Chrome} codebase (over 35~million lines
of C/C++), the same synthesis loop, driven by the open-weight
LLM model Kimi~K2.5~\cite{team2025kimik25}, discovers \textbf{ten}
previously unknown zero-day vulnerabilities, including two
Critical sandbox-escape CVEs (CVE-2026-5280 and
CVE-2026-6297), all confirmed by the vendor through the
Chrome Vulnerability Reward Program.

\paragraph{Contributions}
\begin{itemize}[leftmargin=*,nosep]
  \item A typed graph DSL that unifies all harness
    dimensions into a single searchable representation,
    with a type system that keeps the enlarged space
    tractable (Section~\ref{sec:preliminaries}).
  \item A feedback-driven optimization loop that reads
    runtime signals from the target (coverage, sanitizer
    output, action traces) to diagnose failures and direct
    the search (Section~\ref{sec:optimization}).
  \item State-of-the-art results on
    TerminalBench-2 and Google~Chrome, including ten
    previously unknown zero-day vulnerabilities, including
    two Critical sandbox escapes (Section~\ref{sec:evaluation}).
\end{itemize}

\section{Motivating Example}
\label{sec:example}

\input{fig-case-study}

Before describing \textsc{AgentFlow} in detail, we walk
through a concrete run of the system on a single target to
show how it iteratively improves the harness using feedback
from the target program.
The target is \texttt{libheif}, a widely used reference implementation of the
HEIF/HEVC image format.
The vulnerability is a heap-buffer-overflow in the color conversion logic
(\texttt{CVE-2020-23109}): the function reads an 8-bit alpha plane but hardcodes a 16-bit copy length (\texttt{width * 2}), causing an out-of-bounds read.
Triggering the overflow requires two conditions: the file must pass
an initial format constraint check, and the image must contain an 8-bit alpha channel to trigger the undersized allocation mismatch.

\paragraph{Iteration~1: the file gets thrown out before the bug
ever runs.}
\textsc{AgentFlow} starts with one generic agent, prompted with
``find a memory-safety bug in the HEIF parsing module.''
The agent generates an 815-byte HEIF file, hands it to the
\texttt{libheif} binary, and the binary exits cleanly with no
crash.
The binary writes a line to its standard error stream
saying that the input failed an early format check, and \textsc{AgentFlow} reads that line through the \texttt{stderr}
channel.
The agent's own trace claims it tested the parser; the
program's stderr shows that the parser rejected the file long
before the vulnerable code ran.
The topology at this stage is a single node
(Figure~\ref{fig:case-study}, Iter\,1).

\paragraph{Iteration~2: the right function runs, but the bug is
in a deeper branch.}
\textsc{AgentFlow} reads the stderr message (``format constraint
not satisfied'') and identifies that the single agent failed
at file construction, not at bug triggering: it does not know
enough about HEIF's on-disk format to write a valid file.
The system splits the work into two agents: an \emph{analyzer} whose
only job is to read the HEIF parser's source code and write a
short note describing what a real HEIF file looks like, and a
\emph{crafter} that turns that note into bytes
(Figure~\ref{fig:case-study}, Iter~2).
On the next trial the two-agent harness produces an 817-byte
file that passes the format check, the parser runs on it, and
again the binary exits without crashing.
This time a different signal is available: line-coverage data
(the ``coverage'' feedback in Figure~\ref{fig:case-study},
Iter\,2), which records exactly
which source lines the parser executed during the run.
The coverage record shows that the parser entered the colour
conversion function but never took the specific branch that
handles 8-bit alpha.
The right function ran; the buggy branch inside it was never
visited.

\paragraph{Iteration~3: edit a real HEIF file instead of
inventing one, retry until something complains, and let a
memory checker do the complaining.}
\textsc{AgentFlow} adds a third agent, a \emph{verifier},
downstream of the crafter, and rewrites two things about the
verifier's prompt.
\textbf{(1) Start from a real file, not from scratch.}
Writing a valid HEIF file byte-by-byte is hard, and the
crafter has been spending most of its budget just keeping the
container syntactically legal.
Instead the verifier is told to take a real HEIF image already
on disk (a ``seed'') and change only the one field that
controls the alpha channel's bit depth, leaving every other
byte alone.
\textbf{(2) Try, look, tweak, try again, until something goes
wrong} (\texttt{+ retry loop}).
Instead of running the binary once and giving up, the verifier
is instructed to flip the alpha-bit-depth field, run the
binary, look at the result, adjust the field again, and repeat
until the binary reports an error.
The system also wires a new feedback channel into this
verifier: AddressSanitizer (the ``sanitizer'' feedback channel
in Figure~\ref{fig:case-study}, Iter\,3), a
memory-error detector that has been compiled into the
\texttt{libheif} binary.
AddressSanitizer watches every memory read and write at run
time and reports an error whenever the program touches a byte
outside the buffer the program itself allocated, even when
the program would otherwise have continued without a visible
crash.
Within a handful of tweaks the verifier finds an alpha-bit-depth
value for which AddressSanitizer reports an out-of-bounds read
at the colour-conversion function's \texttt{memcpy} call.
The harness has produced a small, reproducible HEIF file that
triggers the heap-buffer-overflow found in CVE-2020-23109.

\paragraph{What changed}
Across three iterations, \textsc{AgentFlow} added two agents,
rewired the communication graph, rewrote prompts to start
from a real file and to keep retrying until the binary
complained, and routed each new signal to whichever agent
needed it.

\section{Background}
\label{sec:background}

\subsection{Agents and Harnesses}
\label{sec:bg-pipelines}

An \emph{agent} is a large language model (LLM) equipped
with a natural-language instruction (the \emph{system
prompt}) and a set of \emph{tools} it may call (read files,
compile binaries, run tests).
At each step the model reads its accumulated context, picks
an action, observes the result, and repeats until the task
is done or its context window fills up.
A \emph{harness} orchestrates one or more agents into a
pipeline: it fixes which agents exist, what prompt and tools
each receives, how they pass information to one another, and
when the pipeline retries on failure.
Section~\ref{sec:formalization} formalizes the harness as a
five-component tuple.

\subsection{Runtime Feedback Channels}
\label{sec:bg-feedback}

When an agent runs a target program, the program can produce
several kinds of observable output beyond a simple
pass/fail verdict.
We distinguish four classes of runtime feedback that are
relevant to vulnerability discovery:

\begin{enumerate}[leftmargin=*,nosep]
  \item \textbf{Test verdict}: whether the target's own test
    suite passed or failed on the agent's input.
  \item \textbf{Program stdout/stderr}: the raw text the
    target printed during execution, including error messages,
    warnings, and diagnostic output.
  \item \textbf{Line and branch coverage}: which source-code
    lines and conditional branches the target actually
    executed, obtained from LLVM source-based
    instrumentation~\cite{fioraldi2020afl}.
    Coverage reveals whether the agent's input reached the
    code region where a vulnerability might reside.
  \item \textbf{Sanitizer reports}: runtime error reports from
    AddressSanitizer and
    UndefinedBehaviorSanitizer~\cite{serebryany2012asan},
    which detect memory-safety violations (buffer overflows,
    use-after-free) and undefined-behavior errors even when
    the program does not visibly crash.
\end{enumerate}

\noindent
The first two channels are available in any environment; the
last two require the target to be compiled with
instrumentation.

\subsection{Vulnerability Discovery as an Agent Task}
\label{sec:bg-vulndiscovery}

Vulnerability discovery is a sparse-reward sequential
decision problem: most trials end in failure, and failures
carry little diagnostic information.
A harness that never reached the vulnerable function
produces the same pass/fail verdict as one that reached it
and skipped the error-handling branch where the bug resides.
Runtime feedback channels
(Section~\ref{sec:bg-feedback}) can break this ambiguity:
coverage data reveals whether the vulnerable code was
reached; sanitizer output reveals whether a memory error
occurred.

\subsection{Threat Model}
\label{sec:threat-model}

The system is operated by a security analyst on a target
whose source code and build infrastructure are available.
The analyst configures the target's build to emit the
runtime feedback channels of
Section~\ref{sec:bg-feedback}: at minimum the test verdict
and stdout/stderr, optionally augmented with coverage and
sanitizer instrumentation.
This is the standard precondition for source-level security
analysis (source-level audit, coverage-guided fuzzing,
sanitizer-instrumented CI) and is satisfied by every target
in Section~\ref{sec:evaluation}.
Target output is delivered through structured fields (typed
verdict, typed sanitizer report) and freeform stdout/stderr
enclosed in fixed delimiters, so that adversarially-formed
target output cannot be parsed as instructions to the
system.

\section{Problem Formalization}
\label{sec:preliminaries}

\begin{figure*}[t]
\centering
\footnotesize
\begin{minipage}[t]{0.48\textwidth}
\centering
{\normalsize\bfseries (a)\ \,Abstract syntax}\\[4pt]
{\renewcommand{\arraystretch}{1.05}%
\begin{tabular}{@{}r@{~}c@{~}l@{\quad}l@{}}
\multicolumn{4}{@{}l}{\textit{Programs}} \\
\addlinespace[2pt]
$P$ & $::=$ & $(\mathcal{N}, \mathcal{E})$
                                                  & nodes $\mathcal{N}$, edges $\mathcal{E}$ \\
\addlinespace[5pt]
\multicolumn{4}{@{}l}{\textit{Nodes}} \\
\addlinespace[1pt]
$n$ & $::=$  & $\mathsf{agent}(\rho, \pi, m, \phi)$
                                                  & role label $\rho$, prompt $\pi$,        \\
    &        &                                    & model id $m$, tools $\phi \subseteq \mathcal{T}$ \\
    & $\mid$ & $\mathsf{fanout}(n, k)$            & $k$ parallel copies of $n$
                                                    (\,$k \in \mathbb{N}_+$\,)              \\
\addlinespace[5pt]
\multicolumn{4}{@{}l}{\textit{Edges}} \\
\addlinespace[1pt]
$e$ & $::=$  & $n_1 \to n_2$                      & data edge (carries $n_1.\mathsf{out}$ to $n_2$)        \\
    & $\mid$ & $n_1 \to_{g} n_2$                  & guarded edge (fires only when $n_1$'s outcome is $g$); \\
    &        &                                    & $g \in \{\mathsf{ok}, \mathsf{fail}\}$.\quad
                                                    Surface syntax \texttt{n.on\_g >> m}                   \\
\addlinespace[5pt]
\multicolumn{4}{@{}l}{\textit{feedback channels} ($\mathcal{O}$, emitted by the target)} \\
\addlinespace[1pt]
$\mathcal{O}$ & $\ni$  & $\mathsf{cov}(m)$         & line coverage \\
              & $\mid$ & $\mathsf{branch}(m)$      & branch coverage \\
              & $\mid$ & $\mathsf{san}(\tau)$      & sanitizer report \\
              & $\mid$ & $\mathsf{trace}(n)$       & agent trace \\
              & $\mid$ & $\mathsf{outcome}(P)$     & test outcome \\
\end{tabular}}
\end{minipage}%
\hfill
\begin{minipage}[t]{0.48\textwidth}
\centering
{\normalsize\bfseries (b)\ \,Well-formedness}\\[4pt]
{\itshape Inference rules}\\[2pt]
\resizebox{\linewidth}{!}{$\displaystyle
\begin{array}{c}
\displaystyle
\frac{\mathit{fv}(\pi) \subseteq \mathit{In}(n) \cup \mathcal{O}}
     {\vdash \mathsf{agent}(\rho, \pi, m, \phi) : \mathit{Node}}
     \;\textsc{(T-Agent)} \\[7pt]
\displaystyle
\frac{\vdash n : \mathit{Node} \quad k \in \mathbb{N}_+}
     {\vdash \mathsf{fanout}(n, k) : \mathit{Node}}
     \;\textsc{(T-Fanout)} \\[7pt]
\displaystyle
\frac{\vdash n_1 : \mathit{Node} \quad \vdash n_2 : \mathit{Node} \quad
      n_1.\mathsf{out} \in \mathit{fv}(\pi_{n_2})}
     {\vdash n_1 \to n_2 : \mathit{Edge}}
     \;\textsc{(T-Edge)} \\[7pt]
\displaystyle
\frac{\vdash n_1 : \mathit{Node} \quad \vdash n_2 : \mathit{Node} \quad
      g \in \{\mathsf{ok}, \mathsf{fail}\}}
     {\vdash n_1 \to_{g} n_2 : \mathit{Edge}}
     \;\textsc{(T-Branch)} \\[7pt]
\displaystyle
\frac{\forall n \in \mathcal{N}.\;\exists\, n_0 \in \mathit{Src}(P).\;
      n_0 \rightsquigarrow_{\mathcal{E}} n}
     {\vdash P : \mathit{Conn}}
     \;\textsc{(T-Conn)} \\[7pt]
\displaystyle
\frac{(\forall n.\;\vdash n : \mathit{Node}) \quad
      (\forall e.\;\vdash e : \mathit{Edge}) \quad
      \vdash P : \mathit{Conn}}
     {\vdash P : \mathit{Harness}}
     \;\textsc{(T-Pipe)}
\end{array}
$}
\end{minipage}

\par\vspace{4pt}
\begin{center}
{\itshape\small Notation}\\[2pt]
\scriptsize
\renewcommand{\arraystretch}{1.10}%
\begin{tabular}{@{}r@{\;\;}l@{\qquad\qquad}r@{\;\;}l@{}}
$\mathit{In}(n)$
  & $= \{n'.\mathsf{out} \mid (n' \to n) \in \mathcal{E}\}$
    \,\,(inputs visible to $n$)
& $n_0 \rightsquigarrow_{\mathcal{E}} n$
  & $:$ a directed path $n_0 \to^{*} n$ exists in $\mathcal{E}$ \\
$\mathit{Src}(P)$
  & $= \{n \in \mathcal{N} \mid \nexists\,(n' \to n) \in \mathcal{E}\}$
    \,\,(source nodes)
& $n.\mathsf{out}$
  & $:$ output channel of node $n$ \\
$\mathit{fv}(\pi)$
  & $:$ free template variables of $\pi$
& & \\
\end{tabular}
\end{center}
\caption{The \textsc{AgentFlow} language: abstract syntax (a) and
well-formedness rules (b).  The runtime that executes a well-formed
program is described in Algorithm~\ref{alg:harness-opt} and
Section~\ref{sec:system}.}
\label{fig:dsl-rules}
\end{figure*}

\subsection{Multi-Agent Harnesses}
\label{sec:formalization}

A harness $H$ is a multi-agent pipeline that wraps one or more
language models and mediates their interaction with an environment.
We decompose it into five components,
\[
H \;=\; (\mathcal{A},\, \mathcal{G},\, \Sigma,\, \Phi,\, \Psi),
\]
described below.

\begin{itemize}[leftmargin=*,nosep]
  \item $\mathcal{A}$, the \textbf{agent set}.
    Each agent is a triple $(\text{role}, \pi, m)$.
    The \emph{role} is a short label describing the agent's
    function (e.g.\ ``analyst,'' ``explorer,'' ``verifier'').
    The \emph{system prompt}~$\pi$ is a natural-language
    instruction that tells the model what to do, what format to
    produce, and what information to pay attention to.
    The \emph{LLM model}~$m$ is the large language model
    (LLM) that powers the agent, for example,
    Claude~Opus~4.6~\cite{anthropic2025opus} or
    Kimi~K2.5~\cite{team2025kimik25}.
    Different agents in the same harness may use different
    LLM models.
  \item $\mathcal{G} \subseteq \mathcal{A} \times \mathcal{A}$, the
    \textbf{communication topology}: a directed graph whose edges
    determine which agent's output is visible to which other
    agent, and in what order they run.
  \item $\Sigma$, the \textbf{message schema}:
    for each edge $(a, b) \in \mathcal{G}$, a template determining
    what $a$'s output passes to~$b$.
    Templates may reference feedback channels (test output,
    program stdout/stderr, coverage, sanitizer output) as free
    variables; which agents reference which channels is determined
    by the templates themselves.
  \item $\Phi : \mathcal{A} \to 2^{\mathcal{T}\mathit{ools}}$, the
    \textbf{tool allocation}: which tools each agent may invoke
    (e.g.\ read source files, compile and run the target,
    query a database).
  \item $\Psi$, the \textbf{coordination protocol}: how agents
    are composed: sequentially, in parallel, as a fan-out
    (cloning one agent into $k$ independent copies), or in a
    loop-until-success pattern.
\end{itemize}
A single-agent system is the trivial case $|\mathcal{A}|=1$,
$\mathcal{G}=\varnothing$.

\paragraph{Prior work as special cases}
Each recent agent-design system fixes some components of~$H$ and
searches the rest.
\begin{itemize}[leftmargin=*,nosep]
  \item \textbf{Meta-Harness}~\cite{lee2026metaharness}: partially
    mutates the single agent through prompt and template edits
    ($\Sigma$; fixed $|\mathcal{A}|=1$) with
    $\mathcal{G}=\varnothing$.
  \item \textbf{ADAS}~\cite{hu2024adas}: searches $\mathcal{A}$
    via Python code generation with $\mathcal{G}$ pinned to a
    hierarchical loop.
  \item \textbf{AFlow}~\cite{zhang2024aflow}: searches
    $(\mathcal{G}, \Psi)$ by Monte Carlo tree search over a homogeneous agent pool.
  \item \textbf{MaAS}~\cite{yuan2025maas}: samples $\mathcal{A}$
    from a fixed agent pool, with $\Psi$ hardwired to a
    routing cascade in which each query is handed off through
    a fixed sequence of agents.
\end{itemize}
Each system exposes a different, narrow set of edits,
targeted at the components it searches: Meta-Harness rewrites
prompt and template text; ADAS emits new Python controller
code; AFlow picks workflow operators from a fixed library;
MaAS selects agents from a fixed pool.
\textsc{AgentFlow} instead ranges over \emph{all} five
components ($\mathcal{A}, \mathcal{G}, \Sigma, \Phi, \Psi$)
inside a single typed grammar, so every proposed edit,
whether it adds an agent, rewires the graph, rebinds a
channel, revokes a tool, or changes the retry behaviour, is a
local rewrite of a program in the same language.

\subsection{A Typed DSL for Harnesses}
\label{sec:dsl}

We define a typed, graph-structured domain-specific
language for specifying multi-agent harnesses
(Figure~\ref{fig:dsl-rules}).
A program $P = (\mathcal{N}, \mathcal{E})$ is a labelled directed
graph: $\mathcal{N}$ is the node set (agents, optionally lifted
into parallel families by \textsf{fanout}), and $\mathcal{E}$ the
edge set (directed edges between nodes).
The five-component view from Section~\ref{sec:formalization} maps
directly into this DSL: the agent set $\mathcal{A}$ becomes the
agent nodes in $\mathcal{N}$, the topology $\mathcal{G}$ becomes
$\mathcal{E}$, the per-edge schemas $\Sigma$ are encoded by which
upstream outputs each agent's template references, the tool
allocation $\Phi$ is encoded per agent (each \textsf{agent} node
carries its own tool set $\phi$), and the coordination protocol
$\Psi$ is encoded by graph topology (sequential chains,
fan-out into parallel families, an aggregating downstream agent
that consumes their outputs).
The DSL serves as the concrete search space for the
optimizer: every candidate harness is a well-formed
\textsc{AgentFlow} program.
At a high level, readers only need to keep three ideas in
mind: nodes are agents, edges are dataflow or retry links,
and templates determine which upstream outputs and runtime
feedback streams an agent can see.
Figure~\ref{fig:dsl-example} shows one program and its compiled
topology.

\begin{figure*}[t]
\centering
\footnotesize
\begin{minipage}[t]{0.44\textwidth}
\centering
\textbf{\textsc{AgentFlow} program}\\[2pt]
\begin{tikzpicture}[
  codeblock/.style={
    anchor=north west,
    text width=0.90\textwidth,
    inner xsep=5pt, inner ysep=1pt,
    font=\ttfamily\fontsize{6.5}{7.6}\selectfont,
    align=left,
  },
  lbar/.style={
    line width=2pt, line cap=round,
  },
  slabel/.style={
    font=\sffamily\fontsize{6.5}{7.5}\selectfont\bfseries,
    anchor=north west,
  },
]

\node[slabel, text=cBlue!70!black] (l1) at (0, 0)
  {Agent declarations};
\node[codeblock, below=0pt of l1.south west, anchor=north west]
  (c1) {%
\textcolor{cBlue}{\bfseries analyst}
  = agent(role="\textcolor{cBlue}{analyst}",\\
\quad prompt="read \{\{target\}\}
  via \{\{\textcolor{cTeal}{cov}\}\}",\\
\quad tools=\{read\}, model=M)\\
\textcolor{cBlue}{\bfseries explorer}
  = agent(role="\textcolor{cBlue}{explorer}",\\
\quad prompt="craft from
  \{\{\textcolor{cBlue}{analyst}.out\}\}\\
\qquad\quad\ \ guided by
  \{\{\textcolor{cTeal}{branch}\}\}",\\
\quad tools=\{read,exec\}, model=M)\\
\textcolor{cBlue}{\bfseries validator}
  = agent(role="\textcolor{cBlue}{validator}",\\
\quad prompt="confirm crash in
  \{\{\textcolor{cBlue}{probes}.out\}\}\\
\qquad\quad\ \ \ \ using
  \{\{\textcolor{cTeal}{san}\}\}",\\
\quad tools=\{read,exec\}, model=M)%
};
\draw[lbar, cBlue!50]
  ([xshift=-2pt]c1.north west) -- ([xshift=-2pt]c1.south west);

\node[slabel, text=black!50,
  below=2pt of c1.south west, anchor=north west] (l2) {Topology};
\node[codeblock, below=0pt of l2.south west, anchor=north west]
  (c2) {%
\textcolor{cBlue}{\bfseries probes}
  \ \ = \textbf{fanout}(\textcolor{cBlue}{explorer}, k=8)\\
\textcolor{cBlue}{analyst}
  >> \textcolor{cBlue}{probes}
  >> \textcolor{cBlue}{validator}\\
\textcolor{cBlue}{validator}.on\_fail
  >> \textcolor{cBlue}{analyst}%
};
\draw[lbar, black!35]
  ([xshift=-2pt]c2.north west) -- ([xshift=-2pt]c2.south west);

\begin{scope}[on background layer]
  \node[draw=black!25, rounded corners=3pt, fill=black!2,
        inner xsep=6pt, inner ysep=5pt,
        fit=(l1)(c1)(l2)(c2)] {};
\end{scope}

\end{tikzpicture}
\end{minipage}%
\hfill
\begin{minipage}[t]{0.52\textwidth}
\centering
\textbf{Compiled topology}\\[3pt]
\resizebox{\textwidth}{!}{%
\begin{tikzpicture}[
  A/.style={-{Stealth[length=2mm, width=1.8mm]}, line width=1.1pt,
    draw=black!65},
  Aback/.style={-{Stealth[length=2mm, width=1.8mm]}, line width=1pt,
    draw=cTeal, dashed},
  hd/.style={draw=cBlue!55, fill=cBlue!6, rounded corners=2pt,
    inner sep=4pt, font=\normalsize, minimum width=1.7cm,
    minimum height=0.6cm, align=center},
  op/.style={draw=black!40, fill=black!3, rounded corners=2pt,
    inner sep=4pt, font=\normalsize\itshape,
    minimum width=1.3cm, minimum height=0.6cm, align=center},
  obs/.style={font=\scriptsize\bfseries, text=cTeal!70!black,
    draw=cTeal!55, fill=cTeal!8, rounded corners=1.5pt,
    inner xsep=3pt, inner ysep=1.5pt},
]
\node[hd] (ana) at (0, 0) {analyst};
\node[op] (fan) at (2.4, 0) {fanout};

\node[hd, minimum width=1.4cm] (e1) at (5.0, 1.6)  {explorer$_1$};
\node[hd, minimum width=1.4cm] (e2) at (5.0, 0.7)  {explorer$_2$};
\node[font=\Large, text=black!55] (edots) at (5.0, -0.05) {$\vdots$};
\node[hd, minimum width=1.4cm] (e8) at (5.0, -0.85) {explorer$_8$};

\node[hd] (val) at (8.4, 0) {validator};

\node[obs, above=4pt of ana.north] (obscov) {cov};
\node[obs, above=4pt of e1.north] (obsbr) {branch};
\node[obs, above=4pt of val.north] (obssan) {san};

\draw[A] (ana.east) -- (fan.west);
\draw[A] (fan.east) -- (e1.west);
\draw[A] (fan.east) -- (e2.west);
\draw[A] (fan.east) -- (e8.west);
\draw[A] (e1.east) -- (val.north west);
\draw[A] (e2.east) -- (val.west);
\draw[A] (e8.east) -- (val.south west);

\draw[Aback]
  (val.south) .. controls +(-90:1.7) and +(-90:1.7) .. (ana.south)
  node[font=\small\itshape, pos=0.5, above=1pt, text=black!60,
       fill=white, inner xsep=3pt, inner ysep=1pt]
  {validator.on\_fail};

\coordinate (retryBot) at ([yshift=-20pt]$(ana.south)!0.5!(val.south)$);

\begin{scope}[on background layer]
  \node[draw=black!25, rounded corners=3pt, fill=black!2,
        inner xsep=8pt, inner ysep=6pt,
        minimum height=5.6cm,
        fit=(ana)(val)(e1)(e8)(retryBot)(obscov)(obsbr)(obssan)] (diagbox) {};
\end{scope}

\node[anchor=north west, font=\scriptsize, text=black!65]
  at ([yshift=-4pt]diagbox.south west)
  {\tikz[baseline=-0.6ex]{\draw[A,line width=0.9pt] (0,0) -- (5mm,0);}\,data flow
   \quad
   \tikz[baseline=-0.6ex]{\draw[Aback,line width=0.9pt] (0,0) -- (5mm,0);}\,structural feedback
   \quad
   \tikz[baseline=-0.6ex]{\node[obs, inner xsep=2pt, inner ysep=0.5pt]{\strut obs};}\,feedback channel};
\end{tikzpicture}%
}%
\end{minipage}
\caption{An \textsc{AgentFlow} program (left) and its compiled
topology (right).
\textcolor{cBlue}{Blue} names match blue nodes;
\textcolor{cTeal}{teal} channels label the structural signal each
agent reads.
The validator consumes all eight explorer outputs through
\texttt{\{\{probes.out\}\}} and decides which (if any) reproduces
the crash; no separate join operator is needed.}
\Description{Two-panel figure. Left: AgentFlow DSL code with two
color-coded sections separated by colored left bars. The first
section, "Agent declarations" (blue bar), declares three agents --
analyst, explorer, validator -- each with a prompt template that
references both upstream agent outputs (analyst.out, probes.out)
and feedback channels (cov, branch, san) using Jinja-style
double braces. The second section, "Topology" (gray bar), creates
a parallel family probes = fanout(explorer, k=8), wires the
pipeline analyst >> probes >> validator, and adds a retry edge
validator.on_fail >> analyst. Right: a horizontal pipeline of
analyst, fanout, three representative explorers (times eight in
total), and validator. Eight directed edges fan from the
explorers directly into validator (no separate merge node). A
dashed retry arc runs from validator back to analyst. Small teal
badges above each agent label the feedback channel they
read: cov for analyst, branch for explorers, san for validator.}
\label{fig:dsl-example}
\end{figure*}
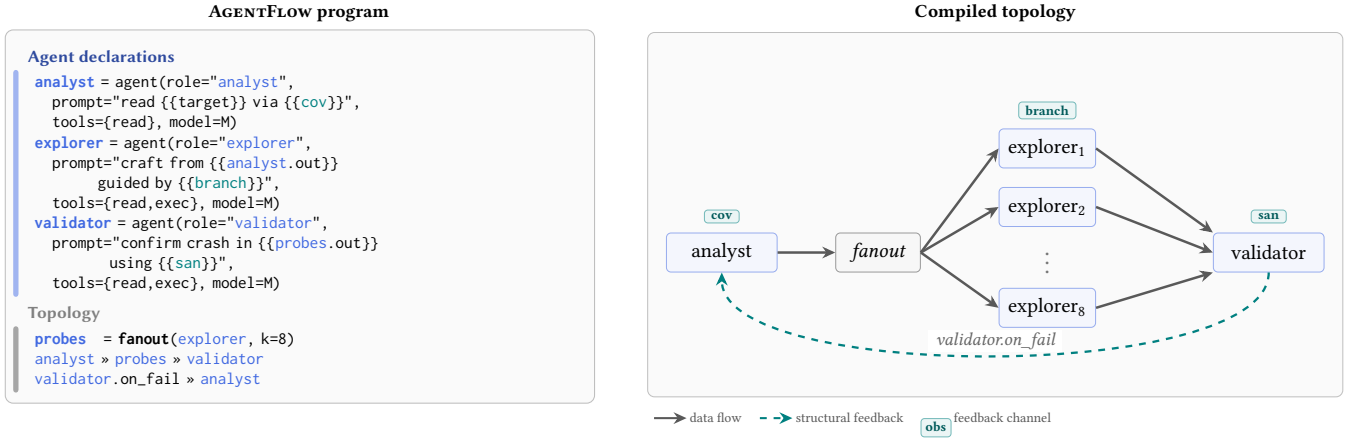

\paragraph{Example}
The program in Figure~\ref{fig:dsl-example} declares three agent
roles.
The \emph{analyst} reads the target source code and produces a
summary of validation logic and reachable code paths.
The \emph{explorer} receives the analyst's output alongside
branch-coverage data and crafts inputs that target specific guard
conditions.
The \emph{validator} receives merged crash reports and sanitizer
output and confirms whether a genuine vulnerability was triggered.
The topology fans out the explorer into eight independent copies,
feeds all eight outputs into the validator, and retries from the
analyst on validation failure.
Each agent's template references only the feedback channels
relevant to its role: the analyst sees coverage, the explorer sees
branch data, and the validator sees sanitizer output.
The three-node sequential chain assumed by most prior work is the
$k{=}1$, retry-free, single-role special case of this program.

\paragraph{Nodes}
The basic node form is
$\mathsf{agent}(\rho, \pi, m, \phi)$: role label~$\rho$,
double-brace template string~$\pi$, model identifier~$m$,
and tool set~$\phi \subseteq \mathcal{T}$.
The role label is a human-readable tag (e.g., \texttt{analyst},
\texttt{explorer}, \texttt{validator}) that identifies the agent in
diagnostics and in the archive.
These labels are task-specific mnemonics rather than a fixed
ontology.
The template~$\pi$ is a Jinja string with free variables drawn from
two sources: upstream node outputs
(written \texttt{\{\{~analyst.out~\}\}}) and feedback channels
(written \texttt{\{\{~cov~\}\}} or \texttt{\{\{~san~\}\}}).
At runtime, the scheduler binds each free variable to concrete
data before dispatching the agent.
The tool set~$\phi$ restricts which tools the agent may invoke;
agents in the same harness may receive disjoint tool sets, so that
a read-only analyst cannot execute the target binary and an
explorer cannot modify source code.

A single structural operator, $\mathsf{fanout}(n, k)$, lifts a
node into a parallel family by cloning $n$ into $k$ independent
copies, each receiving the same upstream context but producing
an independent output.
When a downstream agent's template references the family's
output (e.g.\ \texttt{\{\{probes.out\}\}}), the runtime binds
that variable to a \emph{JSON list} of the $k$ individual
outputs, preserving the order in which they complete.
The complementary join is not a separate constructor: any
downstream agent whose template references the family's outputs
serves as the join point and aggregates them through its own
prompt (e.g., a verifier that selects the most promising
candidate, or an agent that filters outputs whose sanitizer
field reports a crash).

\paragraph{Edges}
An edge $n_1 \to n_2$ declares that $n_1$'s output is in scope
for $n_2$'s template.
The type system (rule \textsc{T-Edge} in Figure~\ref{fig:dsl-rules})
enforces that $n_2$'s template actually references $n_1$'s
output; an edge to an agent that ignores the input is ill-typed
and rejected before execution.
A second edge form, $n_1 \to_{g} n_2$, fires only when $n_1$'s
runtime outcome matches the guard
$g \in \{\mathsf{ok}, \mathsf{fail}\}$ (rule \textsc{T-Branch} in
Figure~\ref{fig:dsl-rules}); the surface syntax
\texttt{n.on\_fail~>>~m} compiles to $n \to_{\mathsf{fail}} m$.
A node may have any number of outgoing edges with mixed forms,
so unconditional successors and guarded branches coexist on the
same node: the example in Figure~\ref{fig:dsl-example} sends the
validator's output forward through one data edge and back to the
analyst through a $\to_{\mathsf{fail}}$ retry edge.
Parallelism is expressed at the node level: a parallel family of
$k$ agents is the single node $\mathsf{fanout}(n, k)$ wired by
ordinary edges, and the join point is whichever downstream agent
references the family's outputs in its template.

\paragraph{feedback channels}
Structural execution signals (test output, program stdout/stderr,
coverage data, and sanitizer reports) are feedback channels.
We write $\mathcal{O}$ for the set of all such channels.
Any node whose template references a channel variable
(e.g., \texttt{\{\{~test\_output~\}\}}, \texttt{\{\{~cov~\}\}},
\texttt{\{\{~san~\}\}}) receives the corresponding data at runtime.
Which agents consume which channels is a property of the templates
in~$\Sigma$, not a separate routing function.
The optimizer controls the allocation by editing templates.
The set~$\mathcal{O}$ adapts to the domain: general
software-engineering tasks expose test results and stdout/stderr;
security targets additionally expose coverage and sanitizer output.

\paragraph{Well-formedness}
Write $\mathit{In}(n)$ for the set of output references from
$n$'s incoming edges, and $\mathit{Src}(P)$ for the nodes with no
incoming edges (the sources).
A well-formed \textsc{AgentFlow} program satisfies three conditions.
First, every node is well-typed: each template's free variables
resolve to upstream outputs or feedback channels, i.e.,
$\mathit{fv}(\pi) \subseteq \mathit{In}(n) \cup \mathcal{O}$
(rule \textsc{T-Agent} in Figure~\ref{fig:dsl-rules}).
Second, every edge $n_1 \to n_2$ has $n_2$'s template actually
referencing $n_1$'s output (rule \textsc{T-Edge}).
Third, the graph is connected: every node is reachable from
some source (rule \textsc{T-Conn}: $\forall n \in \mathcal{N}.\;
\exists\, n_0 \in \mathit{Src}(P).\; n_0
\rightsquigarrow_{\mathcal{E}} n$).
The top-level rule \textsc{T-Pipe} combines all three checks.
In plain language, the validator asks three questions before
any candidate harness is executed:
\begin{itemize}[leftmargin=*,nosep]
  \item Does every prompt reference only data that will
    actually exist at runtime?
  \item Does every declared edge feed information the
    downstream prompt really uses?
  \item Is every node on some path from a source node, so no
    disconnected component is left behind?
\end{itemize}
Figure~\ref{fig:dsl-rules} expresses these same checks as
typing rules.

\paragraph{Runtime}
Execution of a well-formed program follows the DSL itself: the
runtime walks the directed graph, dispatching each node when all
its template variables are bound, and routes the target's
feedback values to every agent whose template references the
corresponding channel.
This is what makes structural feedback available inside the
harness without a dedicated routing component: any agent opts in
to a channel by referencing it
(\texttt{\{\{~test\_output~\}\}}, \texttt{\{\{~cov~\}\}},
\texttt{\{\{~san~\}\}}) in its template.
Algorithm~\ref{alg:harness-opt} (Section~\ref{sec:optimization})
and the implementation in Section~\ref{sec:system} give the
operational details; the formal type system in
Figure~\ref{fig:dsl-rules}(b) is what the outer-loop optimizer
checks before dispatching any candidate harness.

\paragraph{Typing}
The proposer that edits an \textsc{AgentFlow} program is
itself an agent: its proposals are non-deterministic
and routinely malformed (deleted agents still referenced
downstream, renamed nodes whose callers no longer resolve,
edges to agents whose prompt never reads the incoming field,
disconnected components).
Without typing, each such proposal is caught only at dispatch
time, after the harness compiles, the LLM model fires, and the
target runs, burning the largest item in the budget on
programs that were structurally broken from the start.
The three judgements (\textsc{T-Agent}, \textsc{T-Edge},
\textsc{T-Conn}) reduce to a single linear-time graph
traversal that needs no model call, so malformed proposals
are rejected before the scheduler spins up an agent.
Pure-side-effect agents (no output another agent reads)
declare a sentinel boolean status flag, which keeps the
\textsc{T-Edge} check linear-time.

\begin{propbox}
\begin{proposition}[Well-formedness soundness]
\label{prop:executability}
If $\vdash P : \mathit{Harness}$, then no agent in $P$ is ever
dispatched with an unbound template variable: every variable
referenced in any agent's template is bound either by an
upstream node's output or by an feedback channel emitted
by the target.
\end{proposition}
\end{propbox}

\begin{proof}[Proof sketch]
By \textsc{T-Agent}, every agent node $n$ in a well-typed $P$
satisfies
$\mathit{fv}(\pi_n) \subseteq \mathit{In}(n) \cup \mathcal{O}$.
Each variable in $\mathit{In}(n)$ is the output channel of an
upstream node $n' \to n \in \mathcal{E}$, which is bound when
$n'$ completes; the runtime dispatches $n$ only after all such
predecessors have completed, and \textsc{T-Conn} guarantees
that every node is reachable from a source so the dependency
order is well-defined.
Each variable in $\mathcal{O}$ is the name of a feedback
channel emitted by the target on each trial, which is bound
the moment the target produces a value for that channel.
Thus every reference in $\pi_n$ is bound at dispatch time.
\end{proof}

\noindent
Proposition~\ref{prop:executability} says any DSL edit
producing $\vdash P : \mathit{Harness}$ is safe to dispatch
without further structural checks; the outer loop applies the
well-formedness check to every proposed edit before it spends
any inference budget on execution.
The three typing rules are individually standard graph
well-formedness conditions; the contribution is not a
novel type theory but a \emph{practical budget guard}: in our
experiments, roughly $20\%$ of proposer outputs fail the
check and are rejected before consuming any LLM model
inference, which is the dominant cost in the loop.
\textsc{AgentFlow}'s restriction to static topologies (no
runtime agent spawning, no within-execution topology changes)
keeps every candidate harness statically analyzable.

\input{fig-overview-small}

\section{AgentFlow Framework}
\label{sec:optimization}

\begin{algorithm}[t]
\caption{\textsc{HarnessOpt}: Feedback-Guided Harness Synthesis}
\label{alg:harness-opt}
\begin{algorithmic}[1]
\Require Model $\mathcal{M}$, task set $\mathcal{D}$, feedback $\Omega$, score $S$, budget $K$
\Ensure optimized harness $H^{\star}$
\State $\mathcal{X} \gets \varnothing$;\; $d \gets \varnothing$;\; $H^{\star} \gets \varnothing$;\; $s^{\star} \gets 0$
  \Comment{init}
\For{$i = 1, \ldots, K$}
  \State $H \gets \texttt{Propose}_{\mathcal{M}}(d,\, \mathcal{X})$
    \Comment{emit DSL program}
  \State $\{\sigma_T\} \gets \texttt{ExecObserve}(\mathcal{M}, H, \mathcal{D}, \Omega)$
    \Comment{execute \& observe}
  \State $s \gets S\bigl(\{\sigma_T\}\bigr)$
    \Comment{score}
  \If{$s > s^{\star}$}
    \State $H^{\star} \gets H$;\; $s^{\star} \gets s$
  \EndIf
  \State $d \gets \texttt{Diagnose}_{\mathcal{M}}\!\bigl(\{(T,\sigma_T)\}_{T\in\mathcal{D}},\; \mathcal{X}\bigr)$
    \Comment{diagnose}
  \State $\mathcal{X} \gets \mathcal{X} \cup \{(H,\, \{\sigma_T\},\, d)\}$
\EndFor
\State \Return $H^{\star}$
\end{algorithmic}
\end{algorithm}

\paragraph{Optimization objective}
Fix a task set $\mathcal{D}$ and a target environment.
Running $H$ on task~$T$ produces a per-agent trace bundle
$\tau = \{\tau_a\}_{a \in \mathcal{A}}$ and a feedback
bundle $\sigma_T = \Omega(\tau)$.
A domain-specific score function $S$ maps the feedback
bundles to a scalar.
The optimization problem is
\[
H^{\star} \;=\; \arg\max_{H \in \mathcal{H}} \;
  S\!\bigl(\{\sigma_T\}_{T \in \mathcal{D}}\bigr),
\]
where $\mathcal{H}$ is the space of well-formed \textsc{AgentFlow}
programs (Section~\ref{sec:dsl}).
In TerminalBench-2, $S$ is the hidden-test pass rate
$\tfrac{1}{|\mathcal{D}|}\sum_T V(\sigma_T)$;
in the Chrome campaign, $S$ is the count of unique sanitizer
crash signatures.
Because $\mathcal{H}$ is a program space rather than a continuous
parameter space, the optimizer searches it iteratively: at step~$i$
it uses past trials to propose DSL edits that produce $H_{i+1}$.

\subsection{Overview}
\label{sec:opt-overview}

Algorithm~\ref{alg:harness-opt} gives the complete procedure
(Figure~\ref{fig:overview-small}).
The inputs are a LLM model $\mathcal{M}$, a task set
$\mathcal{D}$, a runtime feedback function~$\Omega$
(Section~\ref{sec:bg-feedback}), a domain-specific score
function~$S$, and a step budget~$K$.
Each iteration proceeds through four phases:

\begin{enumerate}[leftmargin=*,nosep]
  \item \textbf{Propose}
    (Section~\ref{sec:opt-propose}).
    An LLM call reads the most recent diagnosis~$d$ together with
    the archive~$\mathcal{X}$ and emits a new \textsc{AgentFlow}
    program~$H$.
    The proposed program must pass the well-formedness check
    (Proposition~\ref{prop:executability}) before dispatch;
    ill-typed proposals are rejected and re-proposed.

  \item \textbf{Execute \& Observe}
    (Section~\ref{sec:opt-execute}).
    The proposed harness~$H$ is dispatched on every task
    $T \in \mathcal{D}$.
    For each task the runtime collects per-agent traces
    $\{\tau_a\}_{a \in H.\mathcal{A}}$ and then applies the
    feedback function~$\Omega$ to obtain a structured
    signal~$\sigma_T$.

  \item \textbf{Score}
    (Section~\ref{sec:opt-execute}).
    A domain-specific function~$S$ reduces the feedback bundle to
    a single number~$s$ (e.g.\ task pass rate or unique crash
    count).
    If $s$ exceeds the incumbent, $H^{\star}$ is updated.

  \item \textbf{Diagnose}
    (Section~\ref{sec:opt-diagnose}).
    An LLM call reads the \emph{full} feedback
    bundle~$\{\sigma_T\}$ (not just the score) alongside each
    agent's action traces, and produces a structured diagnosis~$d$
    identifying which agent failed, why, and what harness edit
    would fix it.
    This diagnosis feeds the next Propose step.
\end{enumerate}

\noindent
After diagnosis, the archive is updated and the
loop returns to Propose.
The loop terminates after $K$ steps and returns~$H^{\star}$.

Because every candidate is a well-formed \textsc{AgentFlow}
program, the search space inherits the structure of the DSL:
the five-component view
$H = (\mathcal{A}, \mathcal{G}, \Sigma, \Phi, \Psi)$ from
Section~\ref{sec:formalization} maps one-to-one onto editable DSL
fields, so a single iteration can simultaneously add an agent
($\mathcal{A}$), rewire the communication graph ($\mathcal{G}$),
update a message template ($\Sigma$), change a tool binding
($\Phi$), or convert a sequential chain into a fan-out
($\Psi$), all expressed as a local rewrite of the same program.
The typing rules from Section~\ref{sec:dsl} act as a
\emph{budget guard}: any structurally broken rewrite is caught
in linear time before the expensive LLM model dispatch.

\subsection{Propose}
\label{sec:opt-propose}

The proposer reads the diagnosis~$d$ and the archive~$\mathcal{X}$,
and emits a new \textsc{AgentFlow} program~$H_{i+1}$.
A single proposal can:
\begin{itemize}[leftmargin=*,nosep]
  \item add or remove an agent ($\mathcal{A}$),
  \item rewire a communication edge or add a retry
    back-edge ($\mathcal{G}$),
  \item rewrite an agent's prompt template or rebind an
    feedback channel ($\Sigma$),
  \item add or restrict an agent's tool set ($\Phi$),
  \item change the coordination protocol, e.g.\ convert a
    sequential chain to a fan-out/merge
    ensemble ($\Psi$).
\end{itemize}

\noindent
The proposer selects which of these to apply by reading the
diagnosis.
A bottleneck localized to a single agent's strategy is typically
answered by a prompt rewrite ($\Sigma$) or tool-set change
($\Phi$); a missing capability is answered by adding a specialist
agent ($\mathcal{A}$) and wiring it into the topology
($\mathcal{G}$); a systematic coordination failure (e.g.\ agents
duplicating work or submitting conflicting outputs) is answered by
a protocol change ($\Psi$).
The archive~$\mathcal{X}$ provides historical context so the
proposer avoids re-proposing edits that were tried and failed in
earlier iterations.

The proposer modifies the harness, not the target: it never edits
source code, generates test inputs, or manipulates the grading
infrastructure.
This is enforced by the DSL: the proposer's output
must parse as a well-formed \textsc{AgentFlow} program, and the
DSL has no constructs for modifying the target.
Each candidate~$H_{i+1}$ passes a three-stage validation
pipeline before dispatch:
(1)~syntactic parsing of the emitted DSL code,
(2)~the well-formedness check from
Proposition~\ref{prop:executability} (template variables resolve,
edges are referenced, graph is connected), and
(3)~a one-shot smoke test on a single task to catch runtime
failures not visible at the type level (e.g.\ a tool that
immediately errors).
If any stage fails, the proposer is given at most two retries per
iteration before the optimizer falls back to the incumbent
harness.

\subsection{Execute, Observe, and Score}
\label{sec:opt-execute}

Once a proposed harness passes validation, it is dispatched on
every task $T \in \mathcal{D}$.
For each task the runtime collects per-agent traces
$\{\tau_a\}_{a \in H.\mathcal{A}}$ and then applies the
feedback function~$\Omega$ to obtain a structured
signal~$\sigma_T$.
The signal includes the test verdict (pass/fail), program
stdout/stderr, and, on instrumented builds, line-level
coverage and sanitizer reports
(Section~\ref{sec:bg-feedback}).

\paragraph{Score function}
The score function~$S$ reduces the feedback bundle to a single
number:
\[
s \;=\; S\!\bigl(\{\sigma_T\}_{T \in \mathcal{D}}\bigr).
\]
Two instantiations appear in this paper:
\begin{itemize}[leftmargin=*,nosep]
  \item \emph{TerminalBench-2}: $S$ is the hidden-test pass rate
    $\tfrac{1}{|\mathcal{D}|}\sum_T V(\sigma_T)$, where
    $V(\cdot) \in \{0,1\}$ is the hidden-test verdict.
  \item \emph{Chrome}: $S$ is the number of distinct
    vulnerabilities discovered, identified by unique
    AddressSanitizer crash signatures.
    Each unique signature corresponds to a distinct memory-safety
    bug (e.g.\ a heap-buffer-overflow at a specific call site).
    A harness that triggers three distinct crashes scores higher
    than one that triggers the same crash three times.
\end{itemize}

\noindent
The score determines \emph{whether} to keep the new harness:
$H^{\star}$ is
updated only when $s$ strictly exceeds the incumbent.
The full feedback bundle~$\{\sigma_T\}$ is then passed to the
Diagnose step, which reads the detailed signals (coverage maps,
sanitizer output, stderr) to determine \emph{what} to fix next.

\paragraph{Archive}
The archive~$\mathcal{X}$ stores
$(H_i,\, \{\sigma_T\}_{T \in \mathcal{D}},\, d_i)$ triples from
every past iteration (harness, feedback, and diagnosis) so
that both the diagnoser and proposer can consult the full
optimization history.
In practice the archive is managed as a fixed-size window: the
top-scoring iteration and the most recent $w{=}3$ iterations are
stored in full, and older entries are compressed to one-line
summaries.
This window keeps the LLM context tractable while preserving the
most relevant history for the proposer to avoid repeating
unsuccessful edits.
Implementation details of the archive manager, prompt templates,
and the per-iteration LLM calls are deferred to
Section~\ref{sec:system}.

\subsection{Diagnose}
\label{sec:opt-diagnose}

The diagnoser receives a per-task bundle consisting of the task
objective, each agent's action summary (truncated or LLM-summarized
when traces exceed the context window), and all available runtime
feedback channels.
In multi-agent harnesses the diagnoser must also
\emph{attribute responsibility}: given that the harness failed a
task, which agent's behavior most directly explains the gap between
the intended outcome and the observed execution?

The diagnosis is structured as four fields:

\begin{enumerate}[leftmargin=*,nosep]
  \item \textbf{Bottleneck agent}: which agent $a \in \mathcal{A}$
    (or which interaction along an edge in $\mathcal{G}$)
    most directly caused the failure.
  \item \textbf{Intended behavior}: what that agent tried to do,
    as reported in its action trace.
  \item \textbf{Actual execution}: what the target program
    actually did, as reported by the runtime feedback channels
    (test verdict, program stderr, and on instrumented builds,
    line-level coverage and sanitizer reports).
  \item \textbf{Corrective edit}: a natural-language description
    of what harness change would close the gap between intended
    and actual behavior.
\end{enumerate}

\noindent
The corrective-edit field is
\emph{harness-directed}: it describes changes to agent prompts
($\Sigma$), tool bindings ($\Phi$), topology ($\mathcal{G}$), or
coordination protocol ($\Psi$), never changes to the target source
code or test inputs.

The feedback channels directly affect diagnosis quality.
When only the binary pass/fail verdict is available, the diagnoser
can observe \emph{that} the harness failed but not \emph{why}.
When coverage data is also available, the diagnoser can localize
the failure to specific code regions the agents never reached.
When sanitizer output is available, the diagnoser can distinguish
a genuine vulnerability from a benign crash or a false positive.
The ablation study in Section~\ref{sec:rq2} provides empirical
evidence for the value of these richer signals: disabling prompt
edits (which are the primary consumer of diagnostic detail) causes
the largest performance drop.

\begin{figure*}[!t]
\centering
\begin{minipage}[c]{0.46\linewidth}
\centering
\begin{tikzpicture}
\begin{axis}[
  width=\linewidth,
  height=7.2cm,
  xmin=0.5, xmax=20.4,
  ymin=33, ymax=88,
  xlabel={Optimization step},
  ylabel={Highest pass rate so far on TerminalBench-2 (\%)},
  xtick={1,4,7,10,13,16},
  ytick={40,50,60,70,80},
  ymajorgrids=false,
  axis on top,
  clip=false,
  legend style={
    at={(0.98,0.02)},
    anchor=south east,
    fill=white,
    draw=black!15,
    inner sep=2pt,
    row sep=-1pt,
    font=\fontsize{6}{7}\selectfont,
  },
  legend cell align=left,
]

\addplot[draw=cGray!55, dashed, line width=0.35pt, forget plot]
  coordinates {(0.5,60.9) (16.5,60.9)};
\addplot[draw=cGray!55, dashed, line width=0.35pt, forget plot]
  coordinates {(0.5,65.6) (16.5,65.6)};
\addplot[draw=cGray!55, dashed, line width=0.35pt, forget plot]
  coordinates {(0.5,66.9) (16.5,66.9)};
\addplot[draw=cGray!55, dashed, line width=0.35pt, forget plot]
  coordinates {(0.5,69.0) (16.5,69.0)};
\addplot[draw=cGray!55, dashed, line width=0.35pt, forget plot]
  coordinates {(0.5,72.4) (16.5,72.4)};
\addplot[draw=cGray!55, dashed, line width=0.35pt, forget plot]
  coordinates {(0.5,74.6) (16.5,74.6)};
\addplot[draw=cGray!55, dashed, line width=0.35pt, forget plot]
  coordinates {(0.5,77.3) (16.5,77.3)};
\addplot[draw=cGray!55, dashed, line width=0.35pt, forget plot]
  coordinates {(0.5,77.7) (16.5,77.7)};

\addplot[draw=cRed, dashed, line width=0.8pt]
  coordinates {(0.5,76.4) (16.5,76.4)};
\addlegendentry{Meta-Harness}
\addplot[draw=cOrange, dashed, line width=0.8pt]
  coordinates {(0.5,81.4) (16.5,81.4)};
\addlegendentry{ForgeCode}

\node[font=\fontsize{5.5}{6}\selectfont, color=cGray, anchor=west]
  at (axis cs:15.8,60.9) {Claude Code};
\node[font=\fontsize{5.5}{6}\selectfont, color=cGray, anchor=west]
  at (axis cs:15.8,64.6) {Terminus 2};
\node[font=\fontsize{5.5}{6}\selectfont, color=cGray, anchor=west]
  at (axis cs:15.8,67.0) {Crux};
\node[font=\fontsize{5.5}{6}\selectfont, color=cGray, anchor=west]
  at (axis cs:15.8,69.6) {Mux};
\node[font=\fontsize{5.5}{6}\selectfont, color=cGray, anchor=west]
  at (axis cs:15.8,72.4) {Droid};
\node[font=\fontsize{5.5}{6}\selectfont, color=cGray, anchor=west]
  at (axis cs:15.8,74.0) {TongAgents};
\node[font=\fontsize{5.5}{6}\selectfont\bfseries, color=cRed, anchor=west]
  at (axis cs:15.8,76.4) {Meta-Harness};
\node[font=\fontsize{5.5}{6}\selectfont, color=cGray, anchor=west]
  at (axis cs:15.8,78.4) {Terminus-KIRA};
\node[font=\fontsize{5.5}{6}\selectfont, color=cGray, anchor=west]
  at (axis cs:15.8,80.0) {Capy};
\node[font=\fontsize{5.5}{6}\selectfont\bfseries, color=cOrange, anchor=west]
  at (axis cs:15.8,82.0) {ForgeCode};

\addplot[
  color=cBlue,
  line width=1.2pt,
  const plot mark left,
]
coordinates {
  (1,35.2) (2,36.0) (3,42.7) (4,46.1)  (5,64.0)
  (6,73.0) (7,76.4) (8,77.5) (9,79.8)  (10,79.8)
  (11,79.8) (12,84.3) (13,84.3) (14,84.3) (15,84.3)
  (16,84.3)
};
\addlegendentry{\textsc{AgentFlow} (this work)}

\end{axis}
\end{tikzpicture}
\end{minipage}\hfill
\begin{minipage}[c]{0.52\linewidth}
\centering

\setlength{\tabcolsep}{4pt}
\begin{tabular}{@{}lcc@{}}
\toprule
\textbf{Harness}
  & \textbf{Score (\%)} & \textbf{Date} \\
\midrule
\textbf{\textsc{AgentFlow} (this work)}
                                        & \textbf{84.3} & 2026-04-17 \\
\midrule
ForgeCode~\cite{tb2_forgecode}          & 81.4 & 2026-03-12 \\
Capy~\cite{tb2_capy}                    & 77.7 & 2026-03-12 \\
Terminus-KIRA~\cite{tb2_kira}           & 77.3 & 2026-02-22 \\
Meta-Harness~\cite{lee2026metaharness}  & 76.4 & 2026-03-30 \\
TongAgents~\cite{tb2_tongagents}        & 74.6 & 2026-02-22 \\
Droid~\cite{tb2_droid}                  & 72.4 & 2026-02-05 \\
Mux~\cite{tb2_mux}                      & 69.0 & 2026-02-13 \\
Crux~\cite{tb2_crux}                    & 66.9 & 2026-02-23 \\
Terminus~2~\cite{tb2_terminus2}         & 65.6 & 2026-02-06 \\
Claude Code~\cite{tb2_claudecode}       & 60.9 & 2026-02-07 \\
\bottomrule
\end{tabular}
\end{minipage}
\caption{Synthesis trajectory (left) and leaderboard comparison (right) on TerminalBench-2 (89 tasks, Claude~Opus~4.6, snapshot 2026-04-17).}
\label{fig:trajectory}
\end{figure*}

\section{Implementation}
\label{sec:system}

The optimizer is a single \textsc{HarnessOpt} loop
(Algorithm~\ref{alg:harness-opt}) instantiated identically
across both targets in Section~\ref{sec:evaluation};
this section lists the concrete components.

\paragraph{Per-run feedback bundle}
After the harness runs on a task, the runtime collects the
output of whichever feedback channels the target's build
provides (Section~\ref{sec:bg-feedback}) into a typed
bundle.
TerminalBench-2 environments wire up the first two; the
Chrome build in Section~\ref{sec:rq4} wires up all four.
The bundle entries are referenced by name in the
\textsc{AgentFlow} program (Section~\ref{sec:dsl}); adding
a new feedback channel to a target's build is an
out-of-band action that does not require touching the
synthesis algorithm.

\paragraph{Diagnoser, proposer, and archive}
Both the diagnoser and the proposer are LLM calls on the
same LLM model as the harness's inner agents in each
evaluation (Claude~Opus~4.6 for TerminalBench-2;
Kimi~K2.5 for Google~Chrome).
All three use each provider's default sampling and tool-use
schema, with Anthropic prompt cache enabled for the
Claude~Opus~4.6 runs ($71.2\%$ cache-hit rate;
Section~\ref{sec:rq1}); per-iteration prompts, model
identifiers, and CLI arguments ship with the artifact for
byte-reproducibility (Appendix~\ref{app:open-science}).
The diagnoser fills the four fields from
Section~\ref{sec:optimization} from the agents' traces and
the per-run feedback bundle (with one-shot LLM summarisation when
traces exceed the context window).
The archive is a fixed-size window: the top-scoring iteration
and the most recent $w{=}3$ iterations in full, older entries
one-line each.
The proposer reads diagnosis$+$archive and emits $H_{i+1}$ as
\textsc{AgentFlow} DSL; a single call takes 5--30\,s.

\paragraph{Validation, edits, and budget}
Each candidate $H_{i+1}$ passes the linear-time
well-formedness check from Section~\ref{sec:dsl}
(Proposition~\ref{prop:executability}: syntax,
graph-connectivity, and template-variable resolution) plus
a one-shot smoke test before dispatch on the full task set;
the proposer is given at most two retries per iteration.
On the TerminalBench-2 campaign of Section~\ref{sec:rq1},
the well-formedness check rejects approximately
\textbf{$20\%$} of proposer outputs as malformed
\textsc{AgentFlow} programs (measured as a fraction of
total proposer-emitted tokens), saving the much larger
cost of dispatching those candidate harnesses on the
$89$-task evaluation pool.

\section{Evaluation}
\label{sec:evaluation}

We evaluated \textsc{AgentFlow} by conducting a set of
experiments designed to answer the following questions:
\begin{itemize}[leftmargin=*,nosep]
  \item \textbf{RQ1: Effectiveness.}
    How does \textsc{AgentFlow} compare against
    state-of-the-art harnesses, and what does the synthesis
    trajectory look like?
  \item \textbf{RQ2: Ablation study.}
    How much performance is lost when structural edits,
    prompt edits, or tool edits are removed from the full
    search space?
  \item \textbf{RQ3: Real-world impact.}
    Can \textsc{AgentFlow} discover previously unknown
    vulnerabilities in production-quality codebases?
\end{itemize}

\paragraph{Setup}
RQ1--RQ2 evaluate the synthesized harness on
\textbf{Terminal~Bench-2}~\cite{terminalbench2} with
\textbf{Claude~Opus~4.6}~\cite{anthropic2025opus} as the
LLM model, against the ten publicly-ranked Claude~Opus~4.6 entries
on the leaderboard (snapshot 2026-04-17,
Figure~\ref{fig:trajectory}); the shared LLM model isolates
differences in harness design.
The candidate harness at every step is a well-formed
\textsc{AgentFlow} program (Section~\ref{sec:dsl}) and the
optimizer ranges over the full
$(\mathcal{A}, \mathcal{G}, \Sigma, \Phi, \Psi)$ design
space; \textsc{AgentFlow} denotes the final harness produced
by Algorithm~\ref{alg:harness-opt} after the optimization
campaign converges, whose compiled topology is in
Figure~\ref{fig:final-arch}.
All experiments run on a public cloud GPU/CPU pool under
each benchmark's default wall-clock budget; the
optimizer source is open-sourced under
\textsc{AgentFlow}~\cite{agentflowrepo}
(Appendix~\ref{app:open-science}).

\paragraph{Baselines}
\label{sec:baselines}
Figure~\ref{fig:trajectory} lists the publicly-ranked
Claude~Opus~4.6 entries on the TerminalBench-2 leaderboard.
\textbf{ForgeCode}~\cite{forgecode2026site} is the strongest
hand-engineered entry ($81.4\%$): a production-grade
multi-agent harness with domain-specific tool sets and
hand-tuned coordination logic.
\textbf{Meta-Harness}~\cite{lee2026metaharness} ($76.4\%$)
is the strongest prior \emph{synthesized} baseline: it runs an
outer-loop optimizer over a single agent's prompts and tool
bindings but keeps the topology fixed at $|\mathcal{A}|{=}1$.
\textbf{Capy}~\cite{tb2_capy} ($77.7\%$) is a
hand-engineered multi-agent system with a pre-defined role
hierarchy.
\textbf{Claude~Code}~\cite{tb2_claudecode} ($60.9\%$) is the
bare default Anthropic deployment, serving as the baseline for
what the LLM model achieves without harness engineering.
These four span the spectrum from bare model to fully
hand-tuned multi-agent systems.

\paragraph{Benchmarks}
TerminalBench-2 is the public agent-harness leaderboard: a
fixed pool of $89$ long-horizon terminal tasks (code
translation, machine-learning pipeline reconstruction,
distributed-systems setup, cryptanalysis, memory-corruption
analysis, vulnerability remediation, secret recovery), each
graded by a hidden test suite under a fixed wall-clock
budget, with a public ranking of complete systems (harness
plus model) on each LLM model.
Google~Chrome is one of the most extensively audited
open-source C/C++ codebases in the world; we use the public
Chrome Vulnerability Reward Program as the externally
validated ground truth for whether a finding is a previously
unknown bug.

\paragraph{Evaluation protocol}
\label{sec:leaderboard-protocol}
\label{sec:multiseed-protocol}
We adopt the leaderboard standard protocol for
TerminalBench-2~\cite{terminalbench2}: the headline score is
the maximum task pass rate observed across multiple replays of a
fixed harness, so that our reported score is directly comparable
to the other public entries in
Figure~\ref{fig:trajectory}; every row in that table is graded
on the same $89$-task pool under the same LLM model and
wall-clock budget.
The optimizer produces a \emph{single, task-agnostic harness
program}: the same \textsc{AgentFlow} graph, prompts, and tool
bindings are deployed identically on all $89$ tasks at each
evaluation step.
The diagnoser observes per-task outcomes to localize harness
weaknesses, but every proposed edit acts on the shared
program; the DSL contains no per-task conditional branches, so
an edit is retained only if it raises the aggregate score
across all seven task categories simultaneously.
A topology change that helps one category at the expense of
others will not survive the aggregate gate.
The Chrome campaign
(Section~\ref{sec:rq3}) applies the same synthesis loop
and DSL to a wholly different domain, LLM model, and task type,
and produces externally validated findings, thus providing
evidence that the method transfers beyond the benchmark suite.

\subsection{RQ1: Effectiveness on TerminalBench-2}
\label{sec:rq1}

To answer this question, we run \textsc{AgentFlow} on
TerminalBench-2 with Claude~Opus~4.6 as the LLM
model and compare against the ten publicly-ranked harnesses on the leaderboard
(Figure~\ref{fig:trajectory}).
The synthesized \textsc{AgentFlow} harness passes
\textbf{75 of 89} tasks (\textbf{84.3\%})
under the same protocol.
This is the highest score in the public leaderboard snapshot
we evaluate against, $2.9$ percentage points above ForgeCode
and $7.9$ percentage points above Meta-Harness.
The gap over Meta-Harness ($7.9$\,pp, $\approx 7$ tasks) is
more substantial.

\paragraph{Synthesis trajectory}
Figure~\ref{fig:trajectory} (left) plots the running maximum
pass rate across the optimization campaign.
The trajectory climbs from $35.2\%$ to $84.3\%$ through three
phases, each targeting a different layer of harness design.
\textbf{Infrastructure} (Steps~1--5, $+28.8$\,pp): the
diagnoser localizes failures to missing environment guards: tool
bindings ($\Phi$) and coordination-protocol fixes ($\Psi$),
discovered from stdout/stderr feedback.
\textbf{Specialization} (Steps~6--9, $+15.8$\,pp): the proposer
adds specialist sub-agents ($\mathcal{A}$), retry edges
($\mathcal{G}$), and tool refinements ($\Phi$).
\textbf{Ensemble} (Step~12, $+4.5$\,pp to $84.3\%$): the
proposer rewrites the topology into a fan-out/merge ensemble
($\mathcal{G}$, $\Psi$) that runs independent attempts in
parallel and cross-validates before submission.
This final harness is what we report as \textsc{AgentFlow}
(topology in Figure~\ref{fig:final-arch}).
The three phases collectively touch all five formalization
components, illustrating that the DSL can express
cross-component edits within a single search process.

\begin{rqbox}
  \noindent\textbf{Result for RQ1: under the public
  TerminalBench-2 leaderboard protocol, \textsc{AgentFlow}
  reaches $84.3\%$, the highest
  score in the snapshot we evaluate against.
  The synthesis trajectory climbs from $35.2\%$ to $84.3\%$,
  and the edits applied along the way touch all five
  formalization components
  ($\mathcal{A}, \mathcal{G}, \Sigma, \Phi, \Psi$).}
\end{rqbox}

\begin{table}[t]
\centering
\caption{Ablation study for \textsc{AgentFlow}. \checkmark\ indicates the dimension is searchable.}
\label{tab:ablation-study}
\setlength{\tabcolsep}{3.5pt}
\renewcommand{\arraystretch}{1.08}
\begin{tabular}{@{}l ccccc c@{}}
\toprule
\textbf{Variant}
  & $\mathcal{A}$
  & $\mathcal{G}$
  & $\Sigma$
  & $\Phi$
  & $\Psi$
  & \textbf{Score} \\
& \scriptsize agents
  & \scriptsize topology
  & \scriptsize schemas
  & \scriptsize tools
  & \scriptsize coord.
  & \\
\midrule
\textbf{Full (\textsc{AgentFlow})}
  & \checkmark & \checkmark & \checkmark & \checkmark & \checkmark
  & \textbf{84.3} \\
No Structure Search
  & & & \checkmark & \checkmark &
  & 76.4 \\
No Prompt Search
  & \checkmark & \checkmark & & \checkmark & \checkmark
  & 51.8 \\
No Tool Search
  & \checkmark & \checkmark & \checkmark & & \checkmark
  & 71.9 \\
\bottomrule
\end{tabular}
\end{table}

\begin{table*}[t]
    \caption{Ten zero-day vulnerabilities in Google~Chrome discovered
    end-to-end by \textsc{AgentFlow} on \textbf{Kimi~K2.5}, every
    one of them accepted through the Chrome Vulnerability
    Reward Program and confirmed by the vendor. Rows without
    a public CVE are listed under the accepted vendor
    identifier currently attached to the report. Public patch
    dates are shown where available. Citations in the
    patch-date column point to the corresponding Chrome
    stable-channel release bulletins.}
    \label{tab:real-world}
    \centering
    \begin{tabular}{@{}lllll@{}}
    \toprule
    \textbf{Target} & \textbf{Vuln.\ Type} & \textbf{Severity} & \textbf{Identifier} & \textbf{Patch Date} \\
    \midrule
    Chrome / WebCodecs & Use-after-free               & Critical     & CVE-2026-5280 & 2026-03-31~\cite{chrome2026mar31} \\
    Chrome / Proxy     & Use-after-free               & Critical & CVE-2026-6297 & 2026-04-15~\cite{chrome2026cve6297} \\
    Chrome / Network   & Use-after-free               & High     & CVE-2026-4454 & 2026-03-18~\cite{chrome2026mar18} \\
    Chrome / Codecs    & Integer overflow             & High     & CVE-2026-5274 & 2026-03-31~\cite{chrome2026mar31} \\
    Chrome / Rendering & Heap Buffer Overflow                & High     & CVE-2026-4462 & 2026-03-18~\cite{chrome2026mar18} \\
    Chrome / Rendering     & Use-after-free               & High & 494352590 & N/A \\
    Chrome / Rendering     & Heap Buffer Overflow                & High & 493534964 & N/A \\
    Chrome / WebRTC     & Heap Buffer Overflow                & High & 488803429 & N/A \\
    Chrome / WebCodecs     & Heap Buffer Overflow                & High & 488585490 & N/A \\
    Chrome / WebGL     & Inappropriate implementation & Medium   & CVE-2026-5291 & 2026-03-31~\cite{chrome2026mar31} \\
    \bottomrule
    \end{tabular}
    \end{table*}

\subsection{RQ2: Ablation Study}
\label{sec:rq2}

A literal leave-one-dimension-out ablation over
$H = (\mathcal{A}, \mathcal{G}, \Sigma, \Phi, \Psi)$ is not
especially meaningful in this DSL because several dimensions
are coupled by construction.
The agent set $\mathcal{A}$ and topology $\mathcal{G}$ are
tightly coupled
($\mathcal{G} \subseteq \mathcal{A} \times \mathcal{A}$):
adding or deleting an agent necessarily changes the graph.
The coordination protocol $\Psi$ is likewise not independent,
because Section~\ref{sec:dsl} encodes $\Psi$ through graph
structure (sequential chains, fan-out, guarded retry edges),
so freezing topology already freezes most of $\Psi$.
At the other extreme, freezing the message schemas $\Sigma$
globally is too strong, because the per-edge templates are
what make new edges and new agents well-typed in the first
place.
We therefore report the ablation study in terms of three
edit families that match the implementation-level taxonomy:
\textbf{structural edits} (affecting
$\mathcal{A}$, $\mathcal{G}$, $\Psi$),
\textbf{prompt edits} (affecting $\Sigma$), and
\textbf{tool edits} (affecting $\Phi$).
Table~\ref{tab:ablation-study} summarizes the ablation
results.  All three variants are direct reruns of
\textsc{AgentFlow} under identical controlled conditions.
The only difference is which edit family the proposer is
allowed to emit.
\emph{No Structure Search} disables structural edits
(agent additions, deletions, and topology rewiring) so
the optimizer can only modify prompts and tools;
\emph{No Prompt Search} freezes all per-agent prompts;
\emph{No Tool Search} freezes the tool-binding map $\Phi$.
Each constraint is enforced at the validator: any proposed
edit whose diff touches a frozen component is rejected before
evaluation, so the proposer LLM cannot circumvent the
restriction.

\noindent
This does not imply that prompt search alone would reach
$84.3\%$: the \emph{No Prompt Search} variant still benefits
from the initial $H_0$ prompts, and disabling any single
family still loses $7.9$--$32.5$\,pp relative to the full
search.
The three edit families are complementary, not redundant.
Prompt edits are the largest single contributor \emph{given
the multi-agent topology that structural edits built}: a
prompt-only optimizer on a fixed single-agent scaffold (as
in OPRO~\cite{yang2024opro} or DSPy~\cite{khattab2024dspy})
would not have the multi-role, fan-out, and retry structure
over which those prompts are optimized.

\begin{rqbox}
  \noindent\textbf{Result for RQ2: under identical controlled
  conditions, disabling prompt edits ($\Sigma$) produces the largest drop
  ($-32.5$\,pp), while disabling structural
  ($\mathcal{A}$, $\mathcal{G}$, $\Psi$) or tool ($\Phi$)
  edits produces smaller drops ($-7.9$ and
  $-12.4$\,pp respectively).
  We interpret this as evidence that prompt edits carry most of the optimization
  signal, while structural and tool edits contribute additional
  gains on top of a working prompt.}
  \end{rqbox}

\subsection{RQ3: Real-World Impact}
\label{sec:rq3}
\label{sec:rq4}

To stress-test whether \textsc{AgentFlow} generalizes beyond
the TerminalBench-2 benchmark, we ran the same synthesis loop
on the Google~Chrome codebase, one of the world's largest and
most extensively audited open-source C/C++ codebases, spanning
over \textbf{35~million lines of code}.
At this scale, even the relevant subsystems (rendering,
networking, codecs) individually exceed frontier-model context
windows by orders of magnitude, making it infeasible for any
single agent to hold the target's structure in context while
simultaneously crafting and triaging exploit inputs.
All four runtime feedback
channels (Section~\ref{sec:bg-feedback}) are wired into the
target's build.
Every other component (the DSL, the diagnose-and-propose
loop (Algorithm~\ref{alg:harness-opt}), and the archive and
validation infrastructure (Section~\ref{sec:system})) is
unchanged.

\paragraph{LLM model}
Chrome's sheer scale (several million lines, complex
multi-process architecture) makes it impractical to run the
full campaign with Claude~Opus~4.6.
We therefore use \textbf{Kimi~K2.5}~\cite{team2025kimik25},
an open-weight sparse Mixture-of-Experts model ($1$\,T
total parameters, $32$\,B active per forward pass) that is
substantially cheaper per token than
Claude~Opus~4.6~\cite{anthropic2025opus}.
The weights are publicly released under a modified MIT
license.
Kimi~K2.5 is not among the top-ranked frontier
models~\cite{chiang2024chatbotarena,lmsys2026lmarena}, so the
campaign simultaneously tests whether \textsc{AgentFlow} can
drive a mid-tier LLM model to real vulnerabilities in a
production codebase, demonstrating that the framework
generalizes across both model tiers and target domains.

Table~\ref{tab:real-world} lists the resulting disclosures.
\textsc{AgentFlow} produced \textbf{ten} previously
unknown zero-day vulnerabilities in Google~Chrome.
We report this campaign as a case study of real-world
capability, not as a compute-matched comparison against
alternative systems.
All ten were accepted by Chrome VRP and confirmed by the
vendor. Six already carry public CVEs; the remaining four
are shown in Table~\ref{tab:real-world} under the accepted
vendor identifiers attached to those reports.
The two most severe results, CVE-2026-5280 and
CVE-2026-6297, are Critical use-after-free vulnerabilities
that enable sandbox escape from an attacker-controlled page
to code execution on the user's host.

\paragraph{Chrome campaign details}
\label{sec:rq4-cost}
The synthesis loop produced the harness in
Figure~\ref{fig:chrome-arch} (Appendix~\ref{app:open-science}):
seven subsystem-specific analysts, an attack-surface mapper and
strategy planner, parallel explorers distributed across the
seven Chrome subsystems, a four-stage crash-triage
pipeline, and a two-stage validation pipeline with six feedback
loops that drive iterative PoC generation.
The campaign ran uninterrupted for 7 days on a public-cloud
pool of 24 nodes, each provisioned with $8\times$H100 GPUs
(192 H100s total).

\begin{rqbox}
\noindent\textbf{Result for RQ3: \textsc{AgentFlow} discovered
ten previously unknown zero-day vulnerabilities in
Google~Chrome, including two Critical sandbox-escape
vulnerabilities (CVE-2026-5280 and CVE-2026-6297) with Kimi~K2.5.}
\end{rqbox}

\section{Related Work}
\label{sec:related-work}

\paragraph{Prior harness optimizers expose uneven search interfaces}
Existing optimizers differ sharply in how directly they expose the
five-component harness $(\mathcal{A},\mathcal{G},\Sigma,\Phi,\Psi)$
to search
(Table~\ref{tab:coverage}).
Meta-Harness~\cite{lee2026metaharness} rewrites a single agent's
prompts, tool bindings, and context-building code through a
free-form Claude Code session, so the lone agent \emph{specification}
is partially mutated each iteration ($\Sigma$, $\Phi$, and the
prompt-level slice of $\mathcal{A}$), but the team cardinality stays
at $|\mathcal{A}|{=}1$: no edit introduces a second agent role or any
inter-agent coordination, so any failure mode that needs more than
one role is structurally out of reach.
ADAS~\cite{hu2024adas} synthesises agent \emph{bodies} as Python
code, including their prompts and tool calls
($\mathcal{A},\Sigma,\Phi$), but the call graph between agents is
fixed by a hand-written controller; the topology $\mathcal{G}$ and
the coordination protocol $\Psi$ never change during search.
AFlow~\cite{zhang2024aflow} runs Monte Carlo tree search over a small library of
\emph{workflow operators} that fix the agent pool, the tool
allocation, and the message schemas; only the workflow graph
$\mathcal{G}$ (and partially the coordination operators $\Psi$) is
searched.
MaAS~\cite{yuan2025maas} samples agent teams from a fixed
pool and routes each query through a hand-coded cascade of
agents, so $\mathcal{A}$ is searched, $\Phi$ is
partially allocated per query, and the coordination patterns outside
the cascade (e.g.\ a verifier looping back to a generator) cannot be
represented at all.
OPRO~\cite{yang2024opro},
TextGrad~\cite{yuksekgonul2025textgrad}, and
DSPy/MIPRO~\cite{khattab2024dspy,opsahlOng2024mipro} restrict
themselves to prompt tuning around a fixed model and fixed
topology.
Across all of these, the search space leaves at least one of the
five components fixed by construction, and every iteration is graded
by a scalar outcome or by the model's own trace.
\textsc{AgentFlow} treats all five components as a single
typed grammar so each \textsc{AgentFlow} edit can move any
component (an agent, an edge, a schema, a tool binding, or
a coordination operator) the diagnoser identifies as the
locus of the failure;
the diagnoser reads the runtime feedback the target
produces on each trial
(Section~\ref{sec:bg-feedback}), so the proposer's rewrites
are guided by what the target actually executed in addition
to the agents' own self-report.

\paragraph{Multi-agent frameworks ship the topology as a constant}
AutoGen~\cite{wu2023autogen}, MetaGPT~\cite{hong2024metagpt},
CAMEL~\cite{li2023camel}, and ChatDev~\cite{qian2024chatdev}
provide expressive runtimes for cooperative planning and code
generation, but the topology is still hand-specified in the
application code rather than searched by an optimizer.
A user can rewrite the agent graph manually, but the framework does
not itself search over topologies or propose graph edits from task
feedback.
We treat the topology as a first-class variable that the optimizer
mutates each iteration.

\paragraph{Self-improving agents close the loop on themselves}
Reflexion~\cite{shinn2023reflexion} stores verbal self-reflections
on failure traces; Self-Refine~\cite{madaan2023selfrefine}
iterates generate-then-critique with the same model on both sides;
Tree of Thoughts~\cite{yao2024tot} expands chain-of-thought into a
search over the model's own intermediate states.
In every case the feedback is the agent's self-report;
\textsc{AgentFlow} additionally consumes structural channels
from the target (test pass/fail, program stdout/stderr, and,
when available, line/branch coverage and sanitizer reports),
which are emitted by the target rather than by the agent and
so are independent of the agent's own description of what it
did.
Voyager~\cite{wang2024voyager} is the closest in spirit, accumulating
a skill library from real Minecraft execution feedback, but the
feedback drives skill \emph{accumulation} rather than rewriting the
agent architecture.
None of these systems edits the harness itself; the topology, tools,
and message schemas of the underlying agent are constants of the
loop.

\paragraph{Coverage-guided fuzzing and LLM agents for security}
AFL++~\cite{fioraldi2020afl} and its descendants use coverage
and distance metrics to steer \emph{input} mutation for a
single fuzzer; \textsc{AgentFlow} reuses the same
instrumentation infrastructure as one of several runtime
feedback channels its diagnoser reads
(Section~\ref{sec:bg-feedback}).
ChatAFL~\cite{meng2024chatafl}, Fuzz4All~\cite{xia2024fuzz4all},
TitanFuzz~\cite{deng2023titanfuzz}, and the Fang et
al.~\cite{fang2024llmagent,fang2024teams} agent teams put an
LLM in front of a fuzzer or exploit pipeline but fix the
harness around it; AutoCodeRover~\cite{zhang2024autocoderover}
and Agentless~\cite{xia2024agentless} fix the harness
topology (localize-then-patch, generate-then-test).
\textsc{AgentFlow} inverts this: the model is fixed and the
harness is the search variable, rewritten at every
iteration from the runtime feedback the target produces.

\providecommand{\cmark}{\ensuremath{\checkmark}}%
\begin{table}[t]
\centering
\caption{Coverage of the five harness components
$(\mathcal{A},\mathcal{G},\Sigma,\Phi,\Psi)$ by prior
optimizers (Section~\ref{sec:formalization}).
\cmark{} = first-class search variable;
(\cmark) = can vary, but only indirectly through a fixed
protocol or parameterization;
\textsc{---} = fixed by construction.
We mark a component only when the published method directly
manipulates it, or when it clearly varies in the published
instantiation; theoretical code-space expressivity alone does not
count.}
\label{tab:coverage}
\setlength{\tabcolsep}{4pt}
\renewcommand{\arraystretch}{1.05}
\begin{tabular}{@{}l c c c c c@{}}
\toprule
\textbf{System}
  & $\mathcal{A}$
  & $\mathcal{G}$
  & $\Sigma$
  & $\Phi$
  & $\Psi$ \\
& \scriptsize agents
  & \scriptsize topology
  & \scriptsize schemas
  & \scriptsize tools
  & \scriptsize coord. \\
\midrule
Meta-Harness~\cite{lee2026metaharness}
  & ---  & ---  & \cmark & \cmark & --- \\
ADAS~\cite{hu2024adas}
  & \cmark & --- & \cmark & \cmark & --- \\
AFlow~\cite{zhang2024aflow}
  & --- & \cmark & --- & --- & (\cmark) \\
MaAS~\cite{yuan2025maas} & \cmark & --- & --- & (\cmark) & --- \\
OPRO~\cite{yang2024opro}   & --- & --- & \cmark & --- & --- \\
TextGrad~\cite{yuksekgonul2025textgrad}   & --- & --- & \cmark & --- & --- \\
DSPy~\cite{khattab2024dspy}  & --- & --- & \cmark & --- & --- \\
\midrule
\textbf{\textsc{AgentFlow}}
  & \cmark & \cmark & \cmark & \cmark & \cmark \\
\bottomrule
\end{tabular}
\end{table}

\section{Conclusion}
\label{sec:conclusion}

This paper presented \textsc{AgentFlow}, a system for the
automated synthesis of multi-agent harnesses.
\textsc{AgentFlow} contributes a unified typed graph DSL
over the five main harness dimensions together with a
feedback-driven outer loop and cheap structural validation
for candidate edits.
On TerminalBench-2, the synthesized harness reaches
$84.3\%$, the highest score among all
Claude~Opus~4.6 entries on the public leaderboard.
The same synthesis loop, re-run on Chrome with Kimi~K2.5,
yields ten previously unknown zero-day vulnerabilities,
including two Critical sandbox-escape CVEs
(CVE-2026-5280 and CVE-2026-6297).

\begin{figure*}[!b]
  \centering
  \resizebox{0.58\textwidth}{!}{%
  \begin{tikzpicture}[
    P/.style={font=\large\bfseries\sffamily, inner sep=2pt},
    Phase/.style={font=\small\bfseries\sffamily, text=black!55,
                  inner sep=1pt},
    L/.style={font=\footnotesize\sffamily},
    Tag/.style={font=\scriptsize\sffamily, text=black!55, inner sep=1pt},
    A/.style={-{Stealth[length=2.2mm, width=2mm]}, line width=1.0pt,
              draw=black!65},
    Aside/.style={-{Stealth[length=2mm, width=1.7mm]}, line width=0.9pt,
              draw=cTeal, dashed},
    sub/.style={draw=cBlue!60, fill=cBlue!7, rounded corners=3pt,
                inner sep=4pt, font=\small\sffamily,
                minimum width=2.2cm, minimum height=0.7cm, align=center},
    worker/.style={draw=cPurple!60, fill=cPurple!8, rounded corners=4pt,
                inner sep=5pt, font=\small\sffamily,
                minimum width=3.2cm, minimum height=2.2cm, align=center},
    gate/.style={draw=cOrange!60, fill=cOrange!8, rounded corners=3pt,
                inner sep=4pt, font=\small\sffamily,
                minimum width=2.0cm, minimum height=0.65cm, align=center},
    eval/.style={draw=cGreen!60, fill=cGreen!8, rounded corners=3pt,
                inner sep=4pt, font=\small\sffamily,
                minimum width=2.0cm, minimum height=0.7cm, align=center},
    io/.style={draw=black!50, fill=black!3, rounded corners=2pt,
                inner sep=3pt, font=\small\sffamily,
                minimum width=1.6cm, minimum height=0.6cm, align=center},
  ]

  \node[io, fill=black!4, minimum width=1.8cm, minimum height=1.0cm,
        align=center] (task) at (-0.6, 0)
    {Task\\instruction};

  \node[sub] (planner)  at (3.6, 1.55) {Planner};
  \node[sub] (env)      at (3.6, 0.30) {Env\\Analyzer};
  \node[sub] (advisor)  at (3.6, -0.95) {Domain\\Advisor};

  \draw[A] (task.east) -- ([xshift=-2pt]planner.west);
  \draw[A] (task.east) -- ([xshift=-2pt]env.west);
  \draw[A] (task.east) -- ([xshift=-2pt]advisor.west);

  \node[sub, fill=cBlue!12, draw=cBlue!70, minimum width=2.5cm]
      (approach) at (7.0, 0.30) {Approach\\Generator\\\textit{(Plan A,B,C)}};
  \draw[A] (planner.east)  -- (approach.west);
  \draw[A] (env.east)      -- (approach.west);
  \draw[A] (advisor.east)  -- (approach.west);

  \draw[line width=1.0pt, draw=black!65]
      (approach.east) -- ++(0.45,0) coordinate (afork);
  \fill[black!65] (afork) circle (1.4pt);

  \node[worker] (w1) at (11.0, 2.55) {};
  \node[worker] (w2) at (11.0, 0.00) {};
  \node[worker] (w3) at (11.0, -2.55) {};

  \node[anchor=north west, font=\scriptsize\sffamily\bfseries,
        text=cPurple!75!black]
      at ([xshift=3pt,yshift=-3pt]w1.north west) {Worker 1};
  \node[anchor=north west, font=\scriptsize\sffamily\bfseries,
        text=cPurple!75!black]
      at ([xshift=3pt,yshift=-3pt]w2.north west) {Worker 2};
  \node[anchor=north west, font=\scriptsize\sffamily\bfseries,
        text=cPurple!75!black]
      at ([xshift=3pt,yshift=-3pt]w3.north west) {Worker 3};

  \node[anchor=south west, font=\tiny\ttfamily, text=black!60]
      at ([xshift=3pt,yshift=2pt]w1.south west) {/tmp/ws\_0};
  \node[anchor=south west, font=\tiny\ttfamily, text=black!60]
      at ([xshift=3pt,yshift=2pt]w2.south west) {/tmp/ws\_1};
  \node[anchor=south west, font=\tiny\ttfamily, text=black!60]
      at ([xshift=3pt,yshift=2pt]w3.south west) {/tmp/ws\_2};

  \foreach \w/\plan in {w1/Plan A, w2/Plan B, w3/Plan C} {
    \node[font=\footnotesize\sffamily\itshape, text=cPurple!70!black]
        at ([yshift=4pt]\w.center) {\plan};
    \node[draw=black!40, fill=white, rounded corners=1pt,
          font=\tiny\sffamily, inner sep=2pt,
          minimum width=2.7cm, anchor=center]
        at ([yshift=-7pt]\w.center) {Tmux + execute\_commands};
    \node[draw=black!40, fill=white, rounded corners=1pt,
          font=\tiny\sffamily, inner sep=2pt,
          minimum width=2.7cm, anchor=center]
        at ([yshift=-22pt]\w.center) {image\_read tool};
  }

  \draw[A] (afork) |- (w1.west);
  \draw[A] (afork) -- (w2.west);
  \draw[A] (afork) |- (w3.west);

  \node[gate] (clean1) at (15.0, 3.20) {Cleanup};
  \node[gate] (verif1) at (15.0, 2.10) {Verifier};
  \node[gate] (clean2) at (15.0,  0.55) {Cleanup};
  \node[gate] (verif2) at (15.0, -0.55) {Verifier};
  \node[gate] (clean3) at (15.0, -2.10) {Cleanup};
  \node[gate] (verif3) at (15.0, -3.20) {Verifier};

  \draw[A] (w1.east) -- (clean1.west);
  \draw[A] (clean1.south) -- (verif1.north);
  \draw[A] (w2.east) -- (clean2.west);
  \draw[A] (clean2.south) -- (verif2.north);
  \draw[A] (w3.east) -- (clean3.west);
  \draw[A] (clean3.south) -- (verif3.north);

  \draw[Aside]
      (verif1.south) .. controls +(-90:1.55) and +(0:0.55) ..
      ([yshift=-6pt]w1.east)
      node[font=\scriptsize\itshape, pos=0.5, text=black!60,
           fill=white, inner xsep=2pt, inner ysep=0.5pt]
        {verifier.on\_fail};
  \draw[Aside]
      (verif2.south) .. controls +(-90:1.55) and +(0:0.55) ..
      ([yshift=-6pt]w2.east)
      node[font=\scriptsize\itshape, pos=0.5, text=black!60,
           fill=white, inner xsep=2pt, inner ysep=0.5pt]
        {verifier.on\_fail};
  \draw[Aside]
      (verif3.south) .. controls +(-90:1.55) and +(0:0.55) ..
      ([yshift=-6pt]w3.east)
      node[font=\scriptsize\itshape, pos=0.5, text=black!60,
           fill=white, inner xsep=2pt, inner ysep=0.5pt]
        {verifier.on\_fail};

  \node[eval, minimum width=2.4cm, minimum height=2.4cm, align=center]
      (merge) at (18.5, 0.0) {\textbf{Evaluator}\\\small sub-agent\\\small picks winner};

  \draw[A] (verif1.east) .. controls +(0:0.6) and +(180:0.5) .. ([yshift=18pt]merge.west);
  \draw[A] (verif2.east) -- (merge.west);
  \draw[A] (verif3.east) .. controls +(0:0.6) and +(180:0.5) .. ([yshift=-18pt]merge.west);

  \node[io, fill=black!4, minimum width=1.8cm, minimum height=1.0cm,
        align=center]
      (out) at (21.6, 0.0) {Submit\\answer};
  \draw[A, line width=1.3pt] (merge.east) -- (out.west);

  \node[Phase] at (3.6, 2.55) {Phase 1: Context};
  \node[Phase] at (7.0, 2.55) {Phase 1.5: Plans};
  \node[Phase] at (11.0, 4.05) {Phase 2: Parallel execution};
  \node[Phase] at (15.0, 4.05) {Phase 3: Self-check};
  \node[Phase] at (18.5, 4.05) {Phase 4: Merge};

  \coordinate (retryDip) at ([yshift=-22pt]verif3.south);

  \begin{scope}[on background layer]
    \node[draw=black!12, dashed, rounded corners=8pt, line width=0.5pt,
          fill=black!1, inner xsep=10pt, inner ysep=12pt,
          fit=(planner)(advisor)(approach)(w1)(w3)
              (verif1)(verif3)(retryDip)(merge)]
          (pipeline) {};
  \end{scope}

  \node[font=\small\sffamily\bfseries, text=cTeal!70!black,
        anchor=west]
      (sigtitle) at ([yshift=-12pt]pipeline.south west)
      {Structural feedback channels:};
  \node[font=\footnotesize\sffamily, text=black!75, anchor=west]
      at ([xshift=8pt]sigtitle.east)
      {test stdout/stderr $\to$ verifier;\quad
       cleanup-script stderr $\to$ pre-submission cleanup;\quad
       per-worker tmux output $\to$ evaluator};

  \node[anchor=west, font=\scriptsize\sffamily]
      at ([yshift=-32pt]pipeline.south west)
      {\textcolor{black!65}{$\rightarrow$}~data flow\quad
       \textcolor{cTeal}{- - $\rightarrow$}~structural feedback\quad
       \textit{(every coloured box is an LLM agent;
       colours group agents by phase)}};

  \end{tikzpicture}%
  }%
  \caption{Final synthesized \textsc{AgentFlow} harness for
  TerminalBench-2: nine specialised agent roles across five phases,
  with three parallel workspaces merged by an evaluator.
  Dashed teal arrows are structural-feedback channels.}
  \Description{Wide horizontal pipeline figure showing the final
  multi-agent harness. From left to right: a task-instruction box
  fans out to three parallel sub-agents (Planner, Env Analyzer,
  Domain Advisor) in Phase~1. Their outputs feed an Approach
  Generator (Phase~1.5) that produces three plans, which fork into
  three parallel Worker boxes (Phase~2), each labelled with its
  isolated workspace path (\texttt{/tmp/ws\_0}, \texttt{/tmp/ws\_1},
  \texttt{/tmp/ws\_2}) and tools (Tmux execute\_commands,
  image\_read). Each worker feeds a Cleanup gate then a Verifier
  gate (Phase~3); failed verifications loop back into the worker
  via dashed teal arrows (structural feedback). All three verified
  outputs feed an Evaluator sub-agent (Phase~4) that picks the
  winning workspace and forwards it to the Submit-answer sink
  (gray box, same style as the Task-instruction source).}
  \label{fig:final-arch}
\end{figure*}
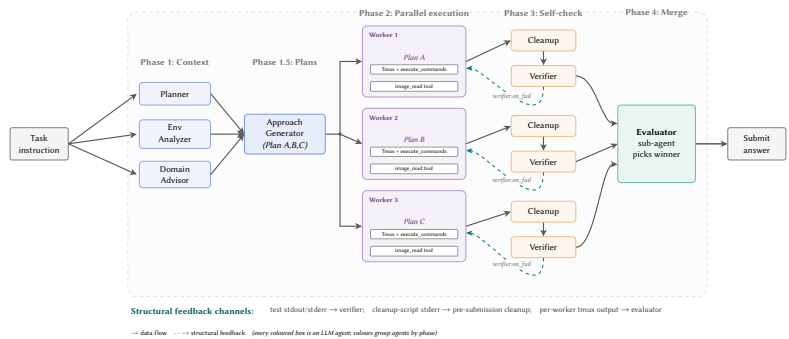
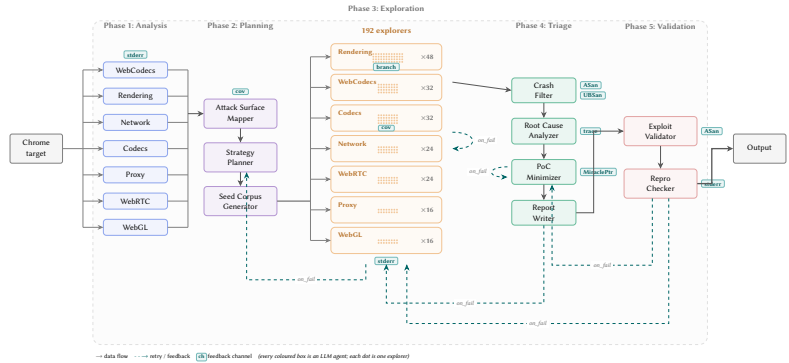
\begin{figure*}[t]
  \centering
  \resizebox{0.58\textwidth}{!}{%
  \begin{tikzpicture}[
    A/.style={-{Stealth[length=2.2mm, width=2mm]}, line width=1.0pt,
              draw=black!65},
    Afan/.style={-{Stealth[length=1.8mm, width=1.5mm]}, line width=0.7pt,
              draw=black!50},
    Aside/.style={-{Stealth[length=2.2mm, width=1.8mm]}, line width=0.9pt,
              draw=cTeal!80!black, dashed},
    analyst/.style={draw=cBlue!60, fill=cBlue!7, rounded corners=3pt,
                inner sep=3pt, font=\small\sffamily,
                minimum width=2.2cm, minimum height=0.55cm, align=center},
    plan/.style={draw=cPurple!60, fill=cPurple!8, rounded corners=3pt,
                inner sep=4pt, font=\small\sffamily,
                minimum width=2.5cm, minimum height=0.65cm, align=center},
    cluster/.style={draw=cOrange!55, fill=cOrange!5, rounded corners=4pt,
                inner xsep=6pt, inner ysep=4pt, font=\small\sffamily,
                minimum width=3.8cm, minimum height=0.90cm, align=center},
    triage/.style={draw=cGreen!60, fill=cGreen!8, rounded corners=3pt,
                inner sep=4pt, font=\small\sffamily,
                minimum width=2.2cm, minimum height=0.60cm, align=center},
    valid/.style={draw=cRed!45, fill=cRed!5, rounded corners=3pt,
                inner sep=4pt, font=\small\sffamily,
                minimum width=2.2cm, minimum height=0.60cm, align=center},
    io/.style={draw=black!50, fill=black!4, rounded corners=2pt,
                inner sep=4pt, font=\small\sffamily,
                minimum width=1.8cm, minimum height=1.0cm, align=center},
    fb/.style={font=\scriptsize\sffamily\bfseries,
               text=cTeal!70!black,
               draw=cTeal!55, fill=cTeal!8, rounded corners=1.5pt,
               inner xsep=3pt, inner ysep=1.5pt},
    Phase/.style={font=\small\bfseries\sffamily, text=black!55,
                  inner sep=1pt},
    clabel/.style={font=\fontsize{7}{8.5}\selectfont\sffamily\bfseries,
                   text=cOrange!80!black, anchor=north west},
    ccount/.style={font=\fontsize{7}{8.5}\selectfont\sffamily\bfseries,
                   text=black!55},
    dot/.style={fill=cOrange!50, circle, inner sep=0pt,
                minimum size=1.8pt},
    rlabel/.style={font=\scriptsize\itshape\sffamily, text=black!55,
                   fill=white, inner xsep=2pt, inner ysep=1pt},
  ]

  \node[io] (src) at (0, 0) {Chrome\\target};

  \def\AX{3.4}
  \node[analyst] (a1) at (\AX,  2.70) {WebCodecs};
  \node[analyst] (a2) at (\AX,  1.80) {Rendering};
  \node[analyst] (a3) at (\AX,  0.90) {Network};
  \node[analyst] (a4) at (\AX,  0.00) {Codecs};
  \node[analyst] (a5) at (\AX, -0.90) {Proxy};
  \node[analyst] (a6) at (\AX, -1.80) {WebRTC};
  \node[analyst] (a7) at (\AX, -2.70) {WebGL};

  \draw[line width=0.7pt, draw=black!50]
      (src.east) -- (1.6, 0);
  \draw[line width=0.7pt, draw=black!50]
      (1.6, 2.70) -- (1.6, -2.70);
  \foreach \a in {a1,a2,a3,a4,a5,a6,a7} {
    \draw[Afan] (1.6, 0 |- \a.west) -- (\a.west);
  }
  \node[fb] at ([yshift=7pt]a1.north) {stderr};

  \def\PX{7.0}
  \node[plan, minimum height=1.0cm] (mapper) at (\PX, 1.2)
      {Attack Surface\\Mapper};
  \node[plan, minimum height=1.0cm] (strat) at (\PX, -0.3)
      {Strategy\\Planner};
  \node[plan, minimum height=1.0cm] (seed) at (\PX, -1.8)
      {Seed Corpus\\Generator};

  \draw[line width=0.7pt, draw=black!50]
      (5.2, 2.70) -- (5.2, -2.70);
  \foreach \a in {a1,a2,a3,a4,a5,a6,a7} {
    \draw[line width=0.7pt, draw=black!50]
        (\a.east) -- (5.2, 0 |- \a.east);
  }
  \draw[A] (5.2, 1.2) -- (mapper.west);
  \draw[A] (mapper.south) -- (strat.north);
  \draw[A] (strat.south) -- (seed.north);
  \node[fb] at ([yshift=7pt]mapper.north) {cov};

  \def\EX{12.0}

  \newcommand{\drawcluster}[6]{%
    \node[cluster, minimum width=3.8cm, minimum height=0.90cm]
        (#1) at (\EX, #2) {};
    \node[clabel] at ([xshift=4pt, yshift=-2pt]#1.north west) {#3};
    \node[ccount, anchor=east]
        at ([xshift=-5pt]#1.east) {$\times$#4};
    \pgfmathsetmacro{\gridw}{(#5-1)*0.09}
    \pgfmathsetmacro{\gridh}{(#6-1)*0.09}
    \pgfmathsetmacro{\ox}{-\gridw/2 + 0.05}
    \pgfmathsetmacro{\oy}{\gridh/2 - 0.04}
    \foreach \r in {1,...,#6} {
      \foreach \c in {1,...,#5} {
        \pgfmathsetmacro{\dx}{\ox + (\c-1)*0.09}
        \pgfmathsetmacro{\dy}{\oy - (\r-1)*0.09}
        \node[dot] at ([xshift=\dx cm, yshift=\dy cm]#1.center) {};
      }
    }
  }

  \drawcluster{c1}{ 3.15}{Rendering}{48}{12}{4}
  \drawcluster{c2}{ 2.10}{WebCodecs}{32}{8}{4}
  \drawcluster{c3}{ 1.05}{Codecs}{32}{8}{4}
  \drawcluster{c4}{ 0.00}{Network}{24}{8}{3}
  \drawcluster{c5}{-1.05}{WebRTC}{24}{8}{3}
  \drawcluster{c6}{-2.10}{Proxy}{16}{8}{2}
  \drawcluster{c7}{-3.15}{WebGL}{16}{8}{2}

  \begin{scope}[on background layer]
    \node[draw=cOrange!40, rounded corners=6pt, fill=cOrange!2,
          inner xsep=10pt, inner ysep=10pt,
          fit=(c1)(c7)] (pool) {};
  \end{scope}
  \node[font=\small\sffamily\bfseries, text=cOrange!80!black,
        anchor=south] at ([yshift=5pt]c1.north) {192 explorers};

  \draw[line width=1.0pt, draw=black!65]
      (seed.east) -- (9.4, -1.8);
  \draw[line width=0.7pt, draw=black!50]
      (9.4, 3.15) -- (9.4, -3.15);
  \foreach \cl in {c1,c2,c3,c4,c5,c6,c7} {
    \draw[Afan] (9.4, 0 |- \cl.west) -- (\cl.west);
  }
  \node[fb] at ([yshift=7pt]c2.north) {branch};
  \node[fb] at ([yshift=7pt]c4.north) {cov};
  \node[fb] at ([yshift=-7pt]c7.south) {stderr};

  \draw[Aside]
      ([yshift=16pt]pool.east)
      .. controls +(0:0.9) and +(0:0.9) ..
      ([yshift=0pt]pool.east)
      node[rlabel, pos=0.5, right=3pt] {on\_fail};

  \def\TX{17.4}
  \node[triage, minimum height=0.85cm] (filter) at (\TX,  2.0)
      {Crash\\Filter};
  \node[triage, minimum height=0.85cm] (dedup)   at (\TX,  0.6)
      {Root Cause\\Analyzer};
  \node[triage, minimum height=0.85cm] (minim)   at (\TX, -0.8)
      {PoC\\Minimizer};
  \node[triage, minimum height=0.85cm] (classif) at (\TX, -2.2)
      {Report\\Writer};

  \draw[A] ([yshift=8pt]pool.east |- filter.west) -- (filter.west);
  \draw[A] (filter.south) -- (dedup.north);
  \draw[A] (dedup.south) -- (minim.north);
  \draw[A] (minim.south) -- (classif.north);

  \node[fb, anchor=west] at ([xshift=4pt, yshift=5pt]filter.east) {ASan};
  \node[fb, anchor=west] at ([xshift=4pt, yshift=-5pt]filter.east) {UBSan};
  \node[fb, anchor=west] at ([xshift=4pt]dedup.east) {trace};

  \draw[Aside]
      ([xshift=-4pt, yshift=5pt]minim.west)
      .. controls +(-180:0.7) and +(-180:0.7) ..
      ([xshift=-4pt, yshift=-5pt]minim.west)
      node[rlabel, pos=0.5, left=3pt] {on\_fail};
  \node[fb, anchor=west] at ([xshift=4pt]minim.east) {MiraclePtr};

  \def\VX{21.4}
  \node[valid, minimum height=1.0cm, minimum width=2.5cm]
      (exploit) at (\VX, 0.6) {Exploit\\Validator};
  \node[valid, minimum height=1.0cm, minimum width=2.5cm]
      (repro) at (\VX, -1.2) {Repro\\Checker};

  \draw[A] (classif.east) -- ++(0.6,0) |- (exploit.west);
  \draw[A] (exploit.south) -- (repro.north);

  \node[fb, anchor=west] at ([xshift=4pt]exploit.east) {ASan};
  \node[fb, anchor=west] at ([xshift=4pt]repro.east) {stderr};

  \node[io] (out) at (24.8, 0) {Output};
  \draw[A, line width=1.3pt] (repro.east) -- ++(0.5,0) |- (out.west);

  \draw[Aside]
      ([xshift=-20pt]pool.south)
      -- ([xshift=-20pt]pool.south |- 0,-4.5)
      -- ([xshift=6pt]strat.south |- 0,-4.5)
         node[rlabel, pos=0.5] {on\_fail}
      -- ([xshift=6pt]strat.south);

  \draw[Aside]
      (classif.south)
      -- (classif.south |- 0,-5.3)
      -- (pool.south |- 0,-5.3)
         node[rlabel, pos=0.5] {on\_fail}
      -- (pool.south);

  \draw[Aside]
      ([xshift=-8pt]repro.south)
      -- ([xshift=-8pt]repro.south |- 0,-4.0)
      -- ([xshift=8pt]minim.south |- 0,-4.0)
         node[rlabel, pos=0.5] {on\_fail}
      -- ([xshift=8pt]minim.south);

  \draw[Aside]
      ([xshift=8pt]repro.south)
      -- ([xshift=8pt]repro.south |- 0,-6.0)
      -- ([xshift=20pt]pool.south |- 0,-6.0)
         node[rlabel, pos=0.5] {on\_fail}
      -- ([xshift=20pt]pool.south);

  \node[Phase] at (\AX,  4.2) {Phase 1: Analysis};
  \node[Phase] at (\PX,  4.2) {Phase 2: Planning};
  \node[Phase] at (\EX,  4.8) {Phase 3: Exploration};
  \node[Phase] at (\TX,  4.2) {Phase 4: Triage};
  \node[Phase] at (\VX,  4.2) {Phase 5: Validation};

  \coordinate (loopDip) at (\EX, -6.3);
  \begin{scope}[on background layer]
    \node[draw=black!12, dashed, rounded corners=8pt, line width=0.5pt,
          fill=black!1, inner xsep=10pt, inner ysep=12pt,
          fit=(a1)(a7)(mapper)(seed)(pool)(c1)(c7)
              (filter)(classif)(loopDip)(exploit)(repro)]
        (pipeline) {};
  \end{scope}

  \node[anchor=west, font=\scriptsize\sffamily]
      at ([yshift=-12pt]pipeline.south west)
      {\textcolor{black!65}{$\rightarrow$}~data flow\quad
       \textcolor{cTeal!80!black}{- - $\rightarrow$}~retry / feedback\quad
       \tikz[baseline=-0.5ex]{\node[fb, inner xsep=2pt, inner ysep=0.5pt,
         font=\fontsize{5.5}{6.5}\selectfont\sffamily\bfseries]
         {\strut ch};}\,feedback channel\quad
       \textit{(every coloured box is an LLM agent;
       each dot is one explorer)}};

  \end{tikzpicture}%
  }%
  \caption{Synthesized \textsc{AgentFlow} harness for the
  Chrome campaign (RQ3): 18~agent roles with
  192~parallel explorers across seven subsystems.
  Six feedback loops drive iterative PoC generation.
  This harness produced the ten zero-days in
  Table~\ref{tab:real-world}.}
  \Description{Wide pipeline diagram of the Chrome harness.}
  \label{fig:chrome-arch}
\end{figure*}

\bibliographystyle{ACM-Reference-Format}
\bibliography{references}

\appendix
\section{Open Science}
\label{app:open-science}

The harness optimizer of Section~\ref{sec:system} (the
\textsc{AgentFlow} runtime, the diagnoser and proposer prompts,
the archive manager, and the proposal-validation pipeline) is
released as open source at
\url{https://github.com/berabuddies/agentflow}~\cite{agentflowrepo},
together with the example pipelines, templates, and CLI used to
drive every optimization round in Section~\ref{sec:evaluation}.
Figure~\ref{fig:final-arch} shows the final synthesized harness
for TerminalBench-2 and
Figure~\ref{fig:chrome-arch} shows the final synthesized harness
for the Chrome campaign.
Full task datasets and per-target proof-of-concept inputs are
not redistributed: the TerminalBench-2 task suite is already
public under its upstream license, and the security-domain
inputs are withheld pending disclosure obligations
(Appendix~\ref{app:ethical}).

\section{Ethical Considerations}
\label{app:ethical}

All previously unknown vulnerabilities reported in this paper
were disclosed to the affected vendor before paper
submission. The ten Chrome vulnerabilities in
Table~\ref{tab:real-world} were filed with Google through the
Chrome Vulnerability Reward Program in Q1~2026 and all ten
were accepted by the vendor. Entries with
public patch metadata currently span releases dated
2026-03-18, 2026-03-31, and 2026-04-15; the
remaining entries are shown in the table under the accepted
vendor identifiers attached to those reports. Each
vulnerability is described in the body at the
level of vulnerability class, affected component, and public
identifier where available.

The released artifact~\cite{agentflowrepo} is the harness
optimizer source. It does not include per-target
proof-of-concept inputs, exploit primitives, crashing
inputs, or trigger conditions for any of the ten Chrome
vulnerabilities, and the per-iteration prompts and
\textsc{AgentFlow} programs released with the artifact are
the configurations from the leaderboard run on
TerminalBench-2 (Section~\ref{sec:rq1}--\ref{sec:rq2}); the
Chrome-specific configuration is held back. We have
exercised the same restraint in the body of the paper.

\end{document}